\documentclass[a4paper,11pt]{article}

\usepackage{fullpage}
\usepackage{times}
\usepackage{soul}
\usepackage{url}
\usepackage[utf8]{inputenc}
\usepackage[small]{caption}
\usepackage{graphicx}
\usepackage[table]{xcolor}
\usepackage{amsmath}
\usepackage{booktabs}
\usepackage{algorithm}
\usepackage{dsfont}
\usepackage{enumerate}
\usepackage{tikz}

\usepackage{amsthm}
\usepackage{amssymb}
\usepackage{algorithm}
\usepackage{algorithmic}
\usepackage{bbm}
\usepackage[T1]{fontenc}

\newcommand{\OPT}{\mathrm{OPT}}
\DeclareMathOperator*{\argmax}{arg\,max}

\usepackage{natbib}
\usepackage{amsmath}
\usepackage{mathtools}
\usepackage{nicefrac}
\usepackage{tcolorbox}
\usepackage{thmtools}
\usepackage{thm-restate}
\usepackage{tikz}
\usepackage[capitalize,noabbrev]{cleveref}

\newcommand{\PM}{W_p}  

\newtheorem{theorem}{Theorem}[section]

\newtheorem{lemma}[theorem]{Lemma}

\newtheorem{claim}{Claim}

\theoremstyle{definition}
\newtheorem{definition}[theorem]{Definition}
\newtheorem{example}[theorem]{Example}

\allowdisplaybreaks

\usepackage{authblk}

\title{\bf The Cost of EFX:  Generalized-Mean Welfare and Complexity Dichotomies with Few Surplus Items}

\author[1]{Eugene Lim}
\author[2]{Tzeh Yuan Neoh}
\author[3]{Nicholas Teh}

\affil[1]{National University of Singapore, Singapore}
\affil[2]{Harvard University, USA}
\affil[3]{University of Oxford, UK}

\date{\vspace{-10mm}}
\begin{document}

\maketitle
\begin{abstract}
    Envy-freeness up to any good (EFX) is a central fairness notion for allocating indivisible goods, yet its existence is unresolved in general. In the setting with few surplus items, where the number of goods exceeds the number of agents by a small constant (at most three), EFX allocations are guaranteed to exist, shifting the focus from existence to efficiency and computation. We study how EFX interacts with generalized-mean ($p$-mean) welfare, which subsumes commonly-studied utilitarian ($p=1$), Nash ($p=0$), and egalitarian ($p \rightarrow -\infty$) objectives. We establish sharp complexity dichotomies at $p=0$: for any fixed $p \in (0,1]$, both deciding whether EFX can attain the global $p$-mean optimum and computing an EFX allocation maximizing $p$-mean welfare are NP-hard, even with at most three surplus goods; in contrast, for any fixed $p \leq 0$, we give polynomial-time algorithms that optimize $p$-mean welfare within the space of EFX allocations and efficiently certify when EFX attains the global optimum. We further quantify the welfare loss of enforcing EFX via the price of fairness framework, showing that for $p > 0$, the loss can grow linearly with the number of agents, whereas for $p \leq 0$, it is bounded by a constant depending on the surplus (and for Nash welfare it vanishes asymptotically). Finally we show that requiring Pareto-optimality alongside EFX is NP-hard (and becomes $\Sigma_2^P$-complete for a stronger variant of EFX). Overall, our results delineate when EFX is computationally costly versus structurally aligned with welfare maximization in the setting with few surplus items.
\end{abstract}

\section{Introduction}
Fair division of indivisible goods is a fundamental problem at the intersection of computer science and economics, modeling diverse real-world problems such as course allocation \citep{budish2012courseallocation}, divorce settlements \citep{BramsTa96}, and inheritance division \citep{pratt1990inheritance} (see also the survey by \citet{amanatidis2023fairdivisionprogress}).

Among the many notions of fairness proposed, \emph{envy-freeness} (EF) and its relaxations remain one of the most widely studied in the field. 
Intuitively, an allocation is EF if every agent values his bundle at least as much as he values any other agent's bundle.
However, in this setting, a complete\footnote{\emph{Completeness} is a common assumption in fair division that mandates that all items must be allocated; otherwise, by leaving all items unallocated, one can obtain a trivially ``fair'' allocation.} envy-free allocation may not exist. 
To address this, weaker notions of EF have been introduced. 
Among the most notable of these relaxations is \emph{envy-freeness up to any good} (EFX) and \emph{envy-freeness up to one good} (EF1).
An allocation is EF1 if any envy that an agent has towards another agent can be removed by deleting \emph{some} good from the latter agent's bundle; it is EFX if the envy can be removed by deleting \emph{any} good.
The latter is strictly stronger and, consequently more desirable as a fairness criterion.
While EF1 allocations always exist and can be computed efficiently for additive valuations~\citep{Budish11,lipton2004approximately}, the existence of EFX allocations remains one of the field's biggest open problems~\citep{amanatidis2023fairdivisionprogress,caragiannis2019unreasonable,Moulin19,procaccia2020enigmatic}.
Positive results so far have been obtained only in (numerous) restricted settings, including (but not limited to) the two-agent setting with general monotone valuations~\citep{plaut2020almost}, identical valuations~\citep{plaut2020almost}, binary valuations~\citep{bu2023efx}, bi-valued instances~\citep{amanatidis2021mnwefx}, and when there are three distinct additive valuation functions in the instance~\citep{prakash_EFX_3types}.\footnote{The last case generalizes the setting with three agents \citep{chaudhury_EFX_3agents} and two distinct additive valuation functions in the instance \citep{mahara2023efx_twoadditive}.}
However, the general existence question for EFX remains unresolved even for additive preferences, and existing algorithmic frameworks fail to extend~\citep{chaudhury2021little,mahara2023extension}.

A particularly promising direction in recent years has focused on the setting with \emph{few surplus items}, where the number of goods ($m$) only slightly exceeds the number of agents ($n$), i.e., $m =n +c$ for a small constant $c$.
This setting is both theoretically and practically appealing.
From a theoretical perspective, it is a special setting where EFX is known to exist even for general monotone valuations, a rare occurrence in the EFX literature.
From a practical perspective, it captures problems relating to the distribution of limited bonus items or small surpluses in shared economies.
\citet{amanatidis2020multiple} first showed that for arbitrary monotone valuations, an EFX allocation always exists when $m \leq n +2$ via an envy cycle elimination procedure.
\citet{mahara2023extension} later extended this guarantee to $m \leq n + 3$ through a potential function-based approach.
\citet{NeohTe25} also focused on this setting but consider the \emph{number} of EFX allocations.
This setting thus provides a general condition ensuring the existence of EFX allocations without further structural restrictions.
As a result, it enables us to move beyond the mere \emph{existence} question and study the \emph{welfare} properties of these EFX allocations.
This motivates our central question: when EFX is guaranteed to exist in the few surplus items regime, what welfare can be achieved under EFX, and what is the computational complexity of finding welfare-optimal EFX allocations?

In most practical applications, fairness is desirable but not the only goal; decision-makers may opt to maximize collective welfare as well.
To study efficiency systematically, we use \emph{generalized-mean} (or \emph{$p$-mean}) welfare, a unifying framework encompassing classical objectives: utilitarian welfare ($p=1$), Nash welfare ($p=0$), and egalitarian welfare ($p \rightarrow -\infty$).
These objectives balance different degrees of inequality aversion, with smaller values of $p$ placing greater emphasis on equity \citep{eckart2024pmeans}.

However, imposing fairness constraints such as EFX introduces a new layer of difficulty. For instance, while EF1 and Pareto-optimal (PO) allocations are known to exist in some settings \citep{barman2018ef1po,caragiannis2019unreasonable},\footnote{An allocation is said to be \emph{PO} if there does not exist another allocation that gives each agent as least as much utility, and some agent strictly greater utility.}  the compatibility of EFX with welfare optimization has remained largely unexplored.
One might expect inherent computational barriers to combining fairness and efficiency, which motivate a deeper investigation into the welfare guarantees of EFX allocations, especially in the few surplus items setting where EFX existence is guaranteed.

To quantify how much welfare is lost by enforcing fairness, the \emph{price of fairness} framework \citep{bertsimas2011priceoffairness,bei2021price,caragiannis2012efficiencyoffairdivision,kurz2014pricesmallnumber} measures the worst-case ratio between the welfare of a (globally) welfare-maximizing allocation and the maximum welfare attainable subject to the fairness constraint.
Prior work has studied this for EF1, EF, and MMS allocations under various welfare objectives such as utilitarian \citep{caragiannis2012efficiencyoffairdivision,bei2021price,kurz2014pricesmallnumber} and egalitarian welfare \citep{celine2023egalitarianPrice}, but no analogous quantitative results are known for EFX in this setting with few surplus items. 
Understanding the ``price of EFX'' in this setting therefore offers a key step towards assessing the trade-off between arguably the strongest known fairness guarantee and collective welfare.

\subsection{Our Contributions}
We initiate the computational analysis of the compatibility between welfare/efficiency and EFX in the few surplus items setting $m = n+c \geq 1$ (with $c \leq 3$),\footnote{Strictly speaking, \emph{few surplus items} might suggest $c \geq 0$. However, for the sake of completeness, we extend our analysis to all $m=n+c \geq 1$, where $c$ can be negative as well.} where EFX (and indeed, EFX$_0$) is guaranteed to exist, shifting the focus from existence to efficiency, computation, and tradeoffs.
We assume additive valuations (as is standard in most works analyzing problems of this nature), and analyze efficiency through generalized-mean ($p$-mean) welfare (for $p \leq 1$), which captures commonly studied welfare notions in fair division, such as utilitarian, Nash, and egalitarian welfare objectives.

In Section~\ref{sec:efx+welfare}, we show a clean complexity dichotomy, with a transition at $p=0$.
For any fixed $p \in (0,1]$, we show that it is NP-hard to both decide whether an EFX allocation can achieve the global $p$-mean optimum and to compute a $p$-mean maximizing EFX allocation.
In contrast, for any fixed $p \leq 0$, we give polynomial-time algorithms that optimize $p$-mean welfare within the space of EFX allocations by enumerating the constant-sized non-singleton bundles and completing the allocation via a matching formulation; we also efficiently certify when EFX can attain the global optimum.

In Section~\ref{sec:price}, we quantify the welfare loss due to imposing EFX by deriving upper and lower bounds on the \emph{price of EFX} as a function of $p$, $n$, and $c$.
For $p > 0$, we show that the welfare loss from enforcing EFX can be linear in $n$; whereas for any fixed $p \leq 0$, the loss is bounded by a constant depending only on $c$ (and for Nash welfare it approaches $1$ as $n$ grows).

In Section~\ref{sec:efx+po}, we explore the compatibility between EFX and a commonly-studied economic efficiency notion, Pareto optimality (PO).
We show that deciding whether an allocation exists that is both EFX and PO is NP-hard (and for a strengthened notion of EFX$_0$, the problem becomes $\Sigma_2^P$-complete).

Finally, in Section~\ref{sec:other_settings}, we provide two complementary extensions that address several pivotal modeling features: (i) assuming agents have strictly positive marginal utilities, we give a polynomial-time algorithm that computes an allocation that is both EFX and maximizes the $p$-mean welfare among all EFX allocations; (ii) we further extend a result in Section~\ref{sec:efx+welfare} by showing that, for any fixed $p \leq 0$ and any constant $c = m-n$, there is an XP (in $c$) algorithm that either computes an EFX allocation maximizing the global $p$-mean welfare (for an arbitrary constant $c$, beyond our main focus on $c \leq 3$) or correctly certifies that no such allocation exists.

For completeness, all our results consider both the standard EFX and a strengthened EFX$_0$ notion.

Together, our results provide the first comprehensive study on the welfare and fairness tradeoffs for a setting where EFX is known to exist. They reveal that despite existence results, achieving welfare guarantees under EFX as a fairness notion remains a nuanced challenge.

\subsection{Related Work}
We discuss related work for several core concepts used or studied in the fair division model with indivisible goods. 

\paragraph{Few surplus items.}
Our work studies the few surplus items setting, where the number of goods only slightly exceeds the number of agents (i.e., $m=n+c$ for small constant $c$). 
This setting has also been studied for other fairness
benchmarks such as \emph{maximin share (MMS)} allocations \citep{kurokawa2016can,bouveret2016conflict,feige2021tight,hsu2024existence}.
For EFX, this setting is particularly attractive because exact existence becomes provable for general monotone
valuations when the surplus is small: EFX is guaranteed for $c\le 2$ \citep{amanatidis2020multiple} and extended to
$c\le 3$ by \citet{mahara2023extension}. 
\citet{NeohTe25} further investigated structural questions in this setting, including counting the number of EFX allocations, as well as exploring several further variants of EFX. 
Our results build on the fact that EFX exists here, and ask what welfare guarantees and computational phenomena remain once existence is no longer the bottleneck.

\paragraph{$p$-mean (generalized-mean) welfare.}
Generalized-mean welfare objectives unify several standard efficiency goals (e.g., utilitarian welfare at $p=1$,
Nash welfare at the $p = 0$, and egalitarian welfare as $p\to -\infty$), and have been extensively studied
algorithmically even without EFX/fairness constraints. 
For general valuations (e.g., subadditive), \citet{barman2020pmeans} give approximation guarantees for $p$-mean welfare, while \citet{barman2021uniformwelfare} obtain stronger constant-factor results under identical valuations. 
\citet{chaudhury2021fairandefficient}  subsequently built on this and provided a polynomial-time algorithm that outputs an allocation satisfying approximations to EFX and achieves an approximation (linear in the number of agents) to the Nash welfare.
For identical additive valuations, \citet{garg2022tractableMNW} provide a polynomial-time approximation scheme.
For structured valuations such as matroid-rank valuations,  \citet{viswanathan2023generalframework} give a strategyproof
mechanism optimizing weighted $p$-means alongside fairness desiderata, and in the same setting \citet{barman2022truthfulmatroid} connect Nash-welfare maximization to (group) strategyproofness. 
\citet{eckart2024pmeans} study normalized $p$-mean objectives and fairness for both goods and chores.

\paragraph{Fairness-efficiency tradeoffs and the price of fairness.}
A broad literature quantifies efficiency loss from fairness constraints via the \emph{price of fairness} framework
\citep{BertsimasEtAl2011,caragiannis2012efficiencyoffairdivision}. 
For indivisible goods, \citet{bei2021price} provide a detailed analysis for EF, EF1, and MMS under utilitarian welfare, while Kurz \citep{kurz2014pricesmallnumber} studies the setting with small items.
Other work considers different objectives (e.g., egalitarian welfare) and additional relaxations (including mixed manna)
\citep{celine2023egalitarianPrice,li2024completelandscapeEF}. For EFX specifically, \citet{bu2025approximability} study the approximability of welfare maximization within fair allocations and, in the setting that allows partial allocations, derive tight $\Theta(n)$ bounds for the utilitarian ``price of EFX''. Our contribution is complementary: we provide explicit welfare guarantees and tight(er) bounds for EFX in the few surplus items setting where \emph{complete} EFX allocations are guaranteed to exist, and we show a sharp qualitative transition around Nash welfare ($p=0$) in both welfare behavior and computational complexity.

\section{Preliminaries} \label{sec:prelims}
Let $N = [n]$ be the set of agents and $G = \{g_1,\dots,g_m\}$ be the set of indivisible goods, where $[k] := \{1,\dots,k\}$ for any positive integer $k$.
A \emph{bundle} refers to a subset of $G$.
Each agent $i\in N$ has a nonnegative \emph{valuation function}  $v_i:2^G\rightarrow\mathbb{R}_{\ge 0}$. We also assume that $v_i$ is \emph{additive}, i.e., $v_i(S) = \sum_{g \in S} v_i(\{g\})$ for any $S\subseteq G$. For notational simplicity, we sometimes write $v_i(g)$ instead of $v_i(\{g\})$ for $g\in G$.
We assume that each good $g \in G$ provides nonzero value to some agent, otherwise every agent has value $0$ for it and we can simply discard it.
For expositional and notational simplicity, we  assume that valuations can be expressed as rational numbers (although we note that this is not necessary, by scaling by common denominators as needed).

A problem \textit{instance} $\mathcal{I} = (N, G, \mathbf{v})$ is defined by the set of agents, goods, and valuation functions.
An \emph{allocation} $\mathcal{A} = (A_1, \dots, A_n)$ is a list of $n$ bundles such that no two bundles overlap, where agent $i$ receives bundle $A_i$; let $\Pi_n(G)$ denote the set of all possible allocations (where obvious, we omit the subscript $n$). Also, let $\Pi_\text{EFX}(G)$ denote the set of allocations that are EFX in a given instance. As standard in the literature, we require allocations to be \emph{complete}, i.e., all goods must be allocated.
An \textit{allocation rule} is a function that maps each instance to an allocation. 

As mentioned in the introduction, EF allocations may not exist.
Our work focus on a popular relaxation, EFX, defined as follows.
The first notion of EFX was conceived by \citet{caragiannis2019unreasonable}, which drops only nonzero-valued goods.
\begin{definition}[EFX] \label{def:EFX}
    An allocation $\mathcal{A} = (A_1,\dots,A_n)$ is \emph{envy-free up to any positively valued item (EFX)} if for all agents $i,j \in N$, either $A_j = \varnothing$ or for all $g \in A_j$ where $v_i(g) > 0$, $v_i(A_i) \geq v_i(A_j \setminus \{g\})$.
\end{definition}

We also consider a stronger variant (introduced by \citet{KYROPOULOU2020110}), which allows for dropping zero-valued goods as well.

\begin{definition}[EFX$_0$] \label{def:EFX_0}
    An allocation $\mathcal{A} = (A_1,\dots,A_n)$ is \emph{envy-free up to any item (EFX$_0$)} if for all agents $i,j \in N$, either $A_j = \varnothing$ or for all $g \in A_j$, $v_i(A_i) \geq v_i(A_j \setminus \{g\})$.
\end{definition}
We study both EFX and the stronger EFX$_0$ notions.
For expositional simplicity, when we refer to EFX without qualification, we mean statements that apply to both notions, unless otherwise stated.
In our technical results, we are explicit about the exact notion under consideration.
This provides a comprehensive treatment of both notions as studied in the literature.

We define the $p$-mean welfare objective, as follows.
\begin{definition}[$p$-mean welfare]
    For any $p \leq 1$, the \emph{$p$-mean welfare} of an allocation $\mathcal{A} = (A_1,\dots, A_n)$ is defined as
    \begin{equation*}
        \PM(\mathcal{A}) := \PM(v_1(A_1),\dots, v_n(A_n)) = \left( \frac{1}{n} \sum_{i \in N} v_i(A_i)^p \right)^{1/p}.
    \end{equation*}
    When $p = 1$, $p= 0$, and $p \rightarrow -\infty$, the $p$-mean welfare corresponds to the utilitarian social welfare, Nash social welfare, and egalitarian welfare, respectively. 
    We denote $W_p^* := \max_{\mathcal{A} \in \Pi_\mathrm{EFX}(G)} W_p(\mathcal{A})$.
    When $p \leq 0$, we adopt the standard convention that $W_p(\mathcal{A}) = 0$ for any allocation $\mathcal{A}$ if $v_i(A_i) = 0$ for some $i \in N$.
\end{definition}

\section{EFX + $p$-Mean Welfare Maximization} \label{sec:efx+welfare}
In this section, we study how EFX interacts with welfare maximization using the previously defined generalized-mean (or $p$-mean) objective $W_p$.
This framework allows us to study how the compatibility of EFX varies with the \emph{degree} of inequality aversion in the welfare measure.

There are two distinct questions one can ask, and they lead to different phenomena and different algorithmic barriers.
The first is a \emph{compatibility} question: does there exist an allocation that is simultaneously EFX and globally $W_p$-optimal among all allocations? When the answer is yes, imposing EFX causes no welfare loss for this objective.
Then, regardless of compatibility, the second question is a \emph{constrained optimization} problem: can we find a $W_p$-optimal allocation among all EFX allocations?
This question remains meaningful whenever EFX allocations are guaranteed to exist (in particular, in our setting with few surplus goods).
We address these questions separately in the following two subsections.

\subsection{Compatibility with Global $p$-Mean Optimality} \label{subsec:global_compatibility}
We first study the following \emph{compatibility} question: does there exist an allocation that simultaneously satisfies EFX and maximizes $p$-mean welfare \emph{globally} among all complete allocations? This is the most stringent way to combine fairness and welfare---a positive answer means that imposing EFX incurs no loss with respect to the chosen objective.

Without fairness constraints, computing (or approximating) $W_p$-optimal allocations has been extensively studied; however, adding fairness can fundamentally change the computational landscape. 
Even for the weaker \emph{envy-freeness up to one good (EF1)} property, the decision problem of whether a utilitarian-welfare maximizing allocation (i.e., $W_1$-optimal) can be chosen to satisfy EF1 is NP-hard \citep{AZIZ2023773}. 
We show that this intractability persists for EFX in our few-surplus items setting: for every fixed $p \in (0,1]$ and any constant surplus $c \le 3$, deciding whether a EFX allocation exists that attains the \emph{global} $p$-mean optimum is NP-hard.

In contrast, when $p \le 0$, a simple structural phenomenon gives us tractability when $m \le n$.
Our results thus reveal a qualitative split at $p=0$.

We begin with the hardness result for $p \in (0,1]$.
\begin{theorem} \label{thm:efx_global_p>0_hard}
    Fix any $p\in(0,1]$ and integer $c\leq 3$. 
    Given an instance $\mathcal{I}=(N,G,\mathbf{v})$ with $m = n + c \geq 1$, determining if there exists an allocation that is both \emph{EFX} and maximizes the $p$-mean welfare among all allocations is NP-hard.
\end{theorem}

Next, we show our result for the case of $p \leq 0$, which neatly contrasts with the above.

\begin{theorem} \label{prop:pleq0_polytime}
    Fix any $p \leq 0$. 
    Given an instance $\mathcal{I} = (N, G, \mathbf{v})$ with $m\leq n$, an allocation that is both \emph{EFX}$_0$ (and hence \emph{EFX}) and maximizes the $p$-mean welfare among all allocations always exists and is polynomial-time computable.
\end{theorem}

The tractability in the above result hinges on the structural simplification that occurs when $m\le n$. 
If the global $W_p$-optimum is positive, then every agent must receive positive value, which forces an allocation with at most one good per agent (and exactly one good per agent when $m=n$). 
Then, EFX$_0$ is immediate: whenever the envied bundle is a singleton, removing its only good leaves the empty bundle.

Once $m>n$, this argument breaks: at least one bundle must contain multiple goods, so EFX constraints become substantive. 
More importantly, even when $p\le 0$, a globally $W_p$-optimal allocation can fail EFX, so alignment with the global $W_p$-optimum may fail outright. 
The next example shows this already for Nash welfare ($p=0$) with just one surplus good.

\begin{example} \label{example:compatibility}
    Let $N=\{1,2,3\}$ and $G=\{g_1,g_2,g_3,g_4\}$, and assume additive valuations given by
    \[
    \begin{array}{c|cccc}
          & g_1 & g_2 & g_3 & g_4\\\hline
    v_1   & 5 & 1 & 0 & 0\\
    v_2   & 0 & 0 & 0 & 5\\
    v_3   & 2 & \tfrac{1}{10} & 1 & 0
    \end{array}
    \]
    Consider the Nash welfare objective ($p=0$). 
    Then the unique Nash welfare-maximizing allocation is $\mathcal{A}^*=(A^*_1,A^*_2,A^*_3)=(\{g_1,g_2\},\{g_4\},\{g_3\})$, but $\mathcal{A}^*$ is not EFX (and hence not EFX$_0$). 
    Consequently, there is no allocation that is
    simultaneously EFX and globally Nash welfare-maximizing.
    
    Due to space constraints, we defer the full explanation to the appendix, and provide just the key idea here.
    Agent~$2$ values only $g_4$, so any Nash-optimal allocation must assign $g_4$ to agent~$2$. To keep agent~$1$'s utility positive, agent~$1$ must receive $g_1$ or $g_2$. If agent~$1$ does not receive $g_1$, then agent 1's utility is at most $1$ and agent~$3$'s utility is at most $v_3(\{g_1,g_3\})=3$, giving Nash product at most $1\cdot 5\cdot 3=15$. The allocation $(\{g_1,g_2\},\{g_4\},\{g_3\})$ achieves utilities $(6,5,1)$ and Nash product $30$, so any Nash-optimal allocation must give $g_1$ to agent~$1$; a quick check of the remaining placements shows $(\{g_1,g_2\},\{g_4\},\{g_3\})$ is the unique Nash maximizer.
    This allocation violates EFX for agent~$3$ toward agent~$1$: removing $g_2$ from agent~$1$ leaves $\{g_1\}$, which agent~$3$ values at $2$, exceeding agent~$3$'s own value $1$ for $\{g_3\}$. Thus, no allocation is both EFX and globally Nash welfare-maximizing.
\end{example}

The above example highlights a new phenomenon that only appears once there is surplus ($m=n+c$ with $c\ge 1$): even on the ``inequality averse'' side ($p\le 0$), the global $W_p$-optimum need not admit any EFX allocation, i.e., Theorem~\ref{prop:pleq0_polytime} does not extend beyond the $m \leq n$ setting.
Nevertheless, when the surplus is small ($c\in\{1,2,3\}$), this remains algorithmically verifiable: we can either construct an EFX allocation that is globally $W_p$-optimal, or certify that none exists.

\begin{theorem} \label{thm:decide-find-efx-global-pmean-maximizer}
    For every fixed $p \le 0$ and $c \in \{1,2,3\}$, given an instance $\mathcal{I} = (N,G,\mathbf{v})$ with $m = n + c \geq 1$, there is a polynomial-time algorithm that either outputs an \emph{EFX}$_0$ (or \emph{EFX}) allocation that also maximizes the $p$-mean welfare among all allocations, or certifies that no such allocation exists.
\end{theorem}

\subsection{Optimization among EFX Allocations} 
When EFX is not compatible with global $p$-mean welfare maximization, the natural fallback is the constrained optimization problem: maximize $p$-mean welfare subject to EFX.
This problem is meaningful even when no globally $p$-mean optimal allocation satisfies EFX, because in our few surplus goods setting, EFX allocations are guaranteed to exist.
The two questions are related but different:
Section~\ref{subsec:global_compatibility} asks whether the welfare optimum can be achieved without sacrificing EFX; here we ask what welfare is achievable given that we insist on EFX.\footnote{This distinction is exactly what later drives the price of EFX bounds we obtain as well: the numerator is the global welfare optimum, and the denominator is the best welfare attainable under EFX.}

Here, we also obtain a clean complexity split at $p=0$: for $p\le 0$ and $c\in \{1,2,3\}$ we can compute a $W_p$-maximizing EFX$_0$ (and hence EFX) allocation in polynomial time, whereas for $p\in(0,1]$, the corresponding optimization problem is NP-hard; as we show with the following results.

\begin{theorem}\label{thm:efx0-efx-welfare-opt}
    Fix any $p \leq 0$ and $c\in\{1,2,3\}$.
    Given an instance $\mathcal{I} = (N,G, \mathbf{v})$ with $m = n + c$, one can in polynomial time compute an \emph{EFX$_0$} (or \emph{EFX}) allocation that maximizes the $p$-mean welfare among all \emph{EFX$_0$} (resp., \emph{EFX}) allocations.
\end{theorem}

\begin{theorem} \label{thm:efx-pmean-opt-hard}
    Fix any $p\in(0,1]$ and integer $c \leq 3$. Given an instance $\mathcal{I} = (N, G, \mathbf{v})$ with $m = n + c \geq 1$, computing an \emph{EFX} (or \emph{EFX}$_0$) allocation that maximizes the $p$-mean welfare among all \emph{EFX} (resp., \emph{EFX}$_0$) allocations is NP-hard.
\end{theorem}

The above show that the computational and structural difficulty of combining EFX with exact welfare maximization is governed by the welfare objective: for $p>0$ the combination is already intractable, while for $p\le 0$ it becomes tractable in sparse instances and remains tractable in our setting when we optimize within the space of EFX allocations.

\section{Price of EFX} \label{sec:price}
The previous section characterized when EFX is compatible with $p$-mean welfare maximization, and when we can (or cannot) efficiently optimize welfare within the space of EFX allocations. 
These results, however, do not answer a more basic quantitative question: when EFX forces us away from the $W_p$-optimum, \emph{how much} does it cost us?
This is especially important in the few surplus items setting, where existence is no longer the bottleneck: for $m \le n+3$, EFX (and EFX$_0$) allocations are guaranteed to exist, so welfare loss is driven purely by the strength of the fairness constraint rather than nonexistence.

To formalize this trade-off, we adopt the \emph{price of fairness} framework, that is well-studied in fair division \citep{bei2021price,bertsimas2011priceoffairness,caragiannis2012efficiencyoffairdivision,celine2023egalitarianPrice,li2024completelandscapeEF}, and study the \emph{price of EFX} (resp., \emph{price of EFX$_0$}), which is the worst-case multiplicative gap between the global $W_p$-optimum and the best $p$-mean welfare achievable by an EFX (resp., EFX$_0$) allocation.
This notion is a natural quantitative complement to the compatibility question from Section~\ref{subsec:global_compatibility}: the price is $1$ exactly on instances where an EFX allocation attains the global $W_p$-optimum, and it grows with the amount of welfare that must be sacrificed to enforce EFX.
Importantly, because EFX$_0$ is a stronger requirement than EFX, its feasible set is smaller; correspondingly, its price can only be larger, and we state our bounds for both notions.

Our bounds reveal a sharp transition at $p=0$, i.e., between objectives that tolerate zero utility ($p > 0$) and objectives that do not ($p \leq 0$).
For $p>0$, the global $W_p$-optimum may concentrate value on a small number of agents, and EFX can force substantial ``deconcentration'' even only when $c\leq 3$ surplus goods are present.
In contrast, for $p\le 0$, any allocation with an agent with zero utility has $p$-mean welfare $0$, so positive welfare necessarily requires giving every agent strictly positive value; here, we show that enforcing EFX incurs only a constant-factor loss depending on the surplus $c$ (and for Nash welfare, the loss vanishes with $n$).

We now formalize these guarantees.

\begin{definition}
    Fix a $p \leq 1$ and an integer $c \leq 3$ where $m = n + c \geq 0$.
    The price of EFX (resp., EFX$_0$) for an instance $\mathcal{I} = (N, G, \mathbf{v})$ is the ratio between the maximum $p$-mean welfare of any allocation for $\mathcal{I}$ (i.e., in $\Pi(G)$)
    and the maximum $p$-mean welfare of an EFX allocation for $\mathcal{I}$ (i.e., in $\Pi_\mathrm{EFX}(G)$):
    \begin{align*}
        \mathrm{PoEFX}^{\mathcal{I}}(p, c) & := \frac{\max_{\mathcal{A}\in \Pi(G)} W_p(\mathcal{A})}{\max_{\mathcal{A}\in \Pi_{\mathrm{EFX}}(G)} W_p(\mathcal{A})}, \quad \mathrm{PoEFX}_0^{\mathcal{I}}(p, c) & := \frac{\max_{\mathcal{A}\in \Pi(G)} W_p(\mathcal{A})}{\max_{\mathcal{A}\in \Pi_{\mathrm{EFX_0}}(G)} W_p(\mathcal{A})},
    \end{align*}
    where the ratio is defined to be $0$ if both the numerator and denominator is $0$.
    The (overall) price of EFX (resp., EFX$_0$) is obtained by taking the supremum over all instances $\mathcal{I}$:
    \begin{align*}
        \mathrm{PoEFX}(p,c) & := \sup_{\mathcal{I}} \mathrm{PoEFX}^{\mathcal{I}}(p,c), \quad 
        \mathrm{PoEFX}_0(p,c) & := \sup_{\mathcal{I}} \mathrm{PoEFX}_0^{\mathcal{I}}(p,c).
    \end{align*}    
\end{definition}
Naturally, it makes sense to study the above quantities only for cases where we know an EFX$_0$ (resp., EFX) allocation exists, otherwise the non-emptiness of $\Pi_{\mathrm{EFX}_0}(G)$ (resp., $\Pi_{\mathrm{EFX}}(G)$) is unclear.
In this regard, for the setting with few surplus items, we know that for $m \leq n+3$, EFX$_0$ (and hence EFX) allocations are guaranteed to exist, and we focus on these cases.

Note that the price depends on three parameters: the number of agents $n$, the surplus $c$, and the parameter $p$.
Theorems~\ref{thm:price_>0_lowerbound} and \ref{thm:price_>0_upperbound} show that for $p\in(0,1]$ the price can grow with the number of agents (already $\Omega(n)$ for constant $c$), while  Theorems~\ref{thm:price_leq0_lowerbound} and \ref{thm:pleq0_upperbound} show that for $p\le 0$ the price is uniformly bounded by $1+c$ (and by $1+c/n$ for Nash welfare).
Taken together, these results quantify when EFX is ``expensive'' (welfare can be polynomially worse) and when it is ``essentially free'' (welfare loss is at most constant, or even asymptotically negligible).

We first consider the case where $p > 0$, and provide the following lower bound.
\begin{theorem} \label{thm:price_>0_lowerbound}
    For any fixed $p \in (0, 1]$ and any integer $c \leq 3$ such that $m=n+c \geq 1$, 
    \begin{equation*}
    	\mathrm{PoEFX}_0(p,c) \geq \mathrm{PoEFX}(p,c) \geq \frac{n+c}{\lfloor (2n+c-1)/n\rfloor} - o(1).
    \end{equation*}
\end{theorem}
Intuitively, the above lower bound is driven by a highly asymmetric instance in which one agent values every good much more than everyone else.
For $p>0$, the welfare objective can benefit from concentrating all goods on that agent, because leaving other agents at (near) zero utility may be optimal.
EFX, however, acts as a strong ``anti-hoarding'' constraint: if a single agent accumulates many goods that others value even slightly, then after removing any good from that bundle, the remainder can still look too attractive and triggers envy.
In the constructed family, this forces the high-value agent's bundle size to be bounded by a function of $(n,c)$ (roughly, a constant when $c$ is constant), which in turn gives us a linear welfare gap between the global $W_p$-optimum and the best EFX allocation.

Next, we provide upper bounds for the same setting.

\begin{theorem} \label{thm:price_>0_upperbound}
    For any $p\in(0,1]$ and any integer $c \leq 3$ with $m=n+c \geq 1$,
    \begin{equation*}
        \mathrm{PoEFX}(p,c) \leq \mathrm{PoEFX}_0(p,c) \leq (n+c) \cdot n^{1/p-1}.
    \end{equation*}
\end{theorem}
The upper bound complements our lower bound by showing that even when EFX is costly, one can always guarantee a nontrivial baseline welfare under EFX$_0$ (hence also under EFX).
This gives us a general worst-case bound in terms of $(n,c,p)$, and explains why the price can scale with $n$ even with constant surplus: the numerator can grow with the total number of goods/agents, whereas the denominator can be certified using only a single highly valuable item.

Next, we consider the case where $p \leq 0$. Our lower bound is as follows.
\begin{theorem} \label{thm:price_leq0_lowerbound}
    For any $p \leq 0$ and any integer $c \leq 3$ with $m =n + c \geq 1$, if $c < 0$ (i.e., $m < n$), then $$\mathrm{PoEFX}(p,c) = \mathrm{PoEFX}_0(p,c) = 0;$$ whereas if $c \in \{0,1,2,3\}$, then 
    \begin{align*}
        \mathrm{PoEFX}_0(p,c) & \geq \mathrm{PoEFX}(p,c) \geq \frac{1}{\max_{\gamma \in \{1,2\}} \Big( \frac{\gamma}{c+1} \cdot \Big(1+\tfrac{c+1-\gamma}{n-1}\Big)^{n-1} \Big)^{1/n}}.
    \end{align*}
\end{theorem}
Intuitively, for $p \leq 0$, any allocation achieving positive welfare must give every agent strictly positive value, ruling out the kind of extreme concentration that drives the $p>0$ lower bound.
In the few surplus items setting, this already imposes a strong structural constraint on welfare maximizers: there are only $c$ ``extra'' goods beyond one per agent, so even the global $W_p$-optimum cannot give too many goods to any single agent without ``starving'' other agents and collapsing welfare to $0$.

Finally, we provide the following upper bounds, also for the case where $p \leq 0$.
\begin{theorem} \label{thm:pleq0_upperbound}
    For any $p \leq 0$ any integer $c \leq 3$ with $m=n+c \geq 1$, if $c < 0$ (i.e., $m < n$), then $$\mathrm{PoEFX}(p,c) = \mathrm{PoEFX}_0(p,c) = 0;$$
    whereas if $c \in \{0,1,2,3\}$, then
    \begin{equation*}
        \mathrm{PoEFX}(p,c) \leq \mathrm{PoEFX}_0(p,c) \leq 1 + c.
    \end{equation*}
    Furthermore, when $p = 0$ and $c \in \{0,1,2,3\}$, then
    $$\mathrm{PoEFX}(0,c) \leq \mathrm{PoEFX}_0(0,c) \leq 1 + \frac{c}{n}.$$
\end{theorem}
The above upper bound shows that this inherent balancing makes EFX comparatively cheap: starting from an global $W_p$-optimizer, one can retain, for each agent, a most valuable item and then allocate the remaining $c$ goods while preserving EFX$_0$ (and hence EFX), losing at most a factor of $1+c$ in $p$-mean welfare.
For Nash welfare ($p=0$), the bound strengthens to $1+c/n$, reflecting an even tighter alignment: with only constant surplus, the geometric mean is already strongly discouraging unequal allocations, so enforcing EFX changes the optimum only marginally as $n$ grows.
Overall, the $p\le 0$ case is one where strong fairness and inequality-averse efficiency objectives are naturally compatible up to small (and often vanishing) losses.

\section{EFX + Pareto-Optimality} \label{sec:efx+po}
Sections~\ref{sec:efx+welfare} and \ref{sec:price} studied quantitative efficiency under EFX through generalized-mean (or $p$-mean) welfare: we asked when EFX is compatible with welfare maximization and, when it is not, how large the worst-case welfare gap can be.
While this perspective is essential for understanding fairness-efficiency tradeoffs, and is the main focus of our work, we want to also briefly consider a more basic desideratum that is standard in fair division (and economics, more broadly): avoid outcomes that are unambiguously wasteful.
In particular, even when an allocation achieves a reasonable welfare guarantee, it may still be \emph{Pareto dominated} (i.e., there could exist another allocation that weakly improves every agent and strictly improves some agent).
Ruling out such dominated outcomes is arguably a weak meaningful notion of economic efficiency, and it is commonly studied in fair division.

This motivates the following question, which is logically orthogonal to the welfare questions above:
\emph{given that EFX allocations are guaranteed to exist in the few surplus items setting, can we also insist that the chosen EFX allocation be Pareto-optimal?}
Importantly, this is a qualitative compatibility question: unlike prior sections, we are not fixing a particular welfare objective.
This is useful both conceptually and practically: PO is objective-free, and it captures the minimal requirement that no value is ``wasted''.
It is defined as follows.
\begin{definition}[PO]
    An allocation $\mathcal{A} = (A_1,\dots,A_n)$ is \emph{Pareto-optimal (PO)} if there does not exist another allocation $\mathcal{A}' = (A'_1,\dots, A'_n)$ such that $v_i(A'_i)\geq v_i(A_i)$ for all $i \in N$, and $v_i(A'_i) > v_i(A_i)$ for some $i \in N$.
\end{definition}
A central theme in fair division is whether fairness and efficiency can be achieved simultaneously.
Even for the weaker EF1 notion, achieving EF1 and PO together is algorithmically delicate; no polynomial-time algorithm is known for general additive valuations (outside restricted domains such as binary valuations~\citep{halpern2020binary}), and the decision problem for EF + PO is already $\Sigma_2^P$-complete~\citep{dekeijzer2009efpo_sigmahard}.
One might hope that moving from EF to the weaker EFX condition makes the picture easier, especially in the few surplus items setting where EFX existence is guaranteed.

Our next result shows that this hope is unfounded, the existence question remains NP-hard in our setting, and in fact the stronger EFX$_0$ variant pushes the problem to the second level of the polynomial hierarchy.

\begin{theorem}\label{thm:PO_EFX}
    Fix any integer $c \leq 3$. 
    Given an instance $\mathcal{I} = (N,G, \mathbf{v})$ with $m=n+c \geq 1$, determining whether there exists:
    (i) an allocation that is both \emph{EFX} and \emph{PO} is NP-hard; and
    (ii) an allocation that is both \emph{EFX$_0$} and \emph{PO} is $\Sigma_2^P$-complete.
\end{theorem}

The above result provides a sharp complement to Sections~\ref{sec:efx+welfare} and \ref{sec:price}.
In our setting, EFX existence is not the bottleneck; nevertheless, refining EFX to also satisfy even the weakest global efficiency notion becomes computationally intractable.
The $\Sigma_2^P$-completeness for EFX$_0$ is especially striking: EFX$_0$ differs from EFX only in whether zero-valued goods are allowed to be removed when certifying lack of envy, yet this small modeling choice elevates the problem beyond NP.
From a modeling perspective, this highlights that zero marginal values can fundamentally change the algorithmic landscape of fairness-efficiency combinations.
From a computational perspective, $\Sigma_2^P$-hardness rules out (under standard assumptions) generic polynomial-time approaches and, in particular, implies that the problem is unlikely to admit a polynomial-size integer linear program encoding.

\section{Other Settings} \label{sec:other_settings}
Our analysis has focused on the few surplus items setting $m=n+c \geq 1$ with $c\le 3$, where EFX (indeed EFX$_0$) allocations are guaranteed to exist and the bottleneck shifts from existence to welfare tradeoffs and computational barriers. The results in Sections~\ref{sec:efx+welfare} to \ref{sec:efx+po} also highlight two pivotal modeling features: surplus items, which makes EFX constraints actually binding, and zero marginal values, which can qualitatively change feasibility and complexity.

In this section, we give two complementary extensions that further isolate these effects. 
In Section~\ref{subsec:NMU} we assume \emph{nonzero marginal utilities} (NMU), under which EFX and EFX$_0$ coincide; this additional structure gives us stronger algorithmic guarantees for welfare optimization within the space of EFX allocations than are possible in general. 
In Section~\ref{subsec:constant_surplus} we return to the compatibility question from Section~\ref{subsec:global_compatibility}: whether the global $W_p$-optimum can be attained subject to EFX, but now from a parameterized perspective: we treat the surplus $c=m-n$ as a parameter and show that for every fixed $p\le 0$ the compatibility decision can be solved in time $(n+m)^{\mathcal{O}(c)}$.

Together, these results help disentangle which phenomena are intrinsic to EFX itself and which arise from sparsity ($m \leq n$) or zero-valued goods, and they further sharpen the boundary between tractable and intractable combinations of fairness and welfare in the few surplus items setting.

\subsection{Nonzero Marginal Utilities} \label{subsec:NMU}
A common technical obstacle in showing the existence and compatibility of EFX in conjunction with welfare/efficiency properties in general (and indeed, in our setting) is the presence of
\emph{zero marginal values} (items that an agent is indifferent to receiving).
Following \citet{plaut2020almost}, we say that a valuation $v_i:2^G\to \mathbb{R}_{\ge 0}$ has \emph{nonzero marginal utilities (NMU)} if for every
bundle $S\subseteq G$ and every good $g\in G\setminus S$, $v_i(S\cup\{g\}) - v_i(S) > 0$.
In particular, given additive valuations this is equivalent to $v_i(g)>0$ for all agents $i$ and goods $g$.
Also note that under NMU, every good is positively valued by every agent, so EFX and EFX$_0$ coincide.

Then, we show that we can compute an EFX (equivalently, EFX$_0$) allocation that maximizes $p$-mean welfare for any $p \leq 1$ in our setting, with the following result.

\begin{theorem} \label{thm:nmu-efx-pmean}
    Fix any $p\le 1$ and integer $c\le 3$. 
    Given an instance $\mathcal{I}=(N,G,\mathbf{v})$ satisfying NMU with $m=n+c\ge 1$,
    an allocation that is both EFX (equivalently, EFX$_0$) and maximizes the $p$-mean welfare among all EFX
    (equivalently, EFX$_0$) allocations always exists and is polynomial-time computable.
\end{theorem}

\subsection{EFX + Global $p$-Mean Optimality with Fixed $c$} \label{subsec:constant_surplus}
Section~\ref{subsec:global_compatibility} resolves the compatibility question in our main setting where $c\in\{1,2,3\}$: for every fixed $p\le 0$, Theorem~3.4 gives a polynomial-time algorithm that either finds an EFX allocation attaining the global $p$-mean optimum or certifies that no such allocation exists. 
We now present a parameterized generalization of this idea to any arbitrary (fixed) surplus $c$.

The point of this parameterized extension is primarily conceptual (and not meant to supplant Theorem~\ref{thm:decide-find-efx-global-pmean-maximizer}): it isolates the underlying reason behind Theorem~\ref{thm:decide-find-efx-global-pmean-maximizer} as to why small surplus enables algorithmic verification on the $p\le 0$ side and shows that a similar approach scales with the surplus parameter. 
This viewpoint lets us certify compatibility even outside the $c\le 3$ setting (where EFX existence is not guaranteed), while keeping the paper's main narrative centered on the setting in which existence is known and welfare tradeoffs are meaningful.

We show that for every fixed $p \leq 0$, the decision problem is XP with respect to $c$, with the following result.

\begin{theorem} \label{thm:other_fixedc}
    Fix any $p \le 0$. For every integer $c$, there is an algorithm running in time $(n+m)^{\mathcal{O}(c)}$ that, given an instance $\mathcal{I}=(N,G,\mathbf{v})$ with $m=n+c \ge 1$ and $\max_{\mathcal{A}\in \Pi(G)} W_p(\mathcal{A}) > 0$, 
    either outputs an EFX (or EFX$_0$) allocation that also maximizes the $p$-mean among all allocations, or certifies that no such allocation exists.
\end{theorem}

We note that an analogous result does not extend to $p \in (0,1]$.
For every fixed $p\in(0,1]$, the compatibility problem is NP-hard already for $c=0$ (see Theorem~\ref{thm:efx_global_p>0_hard}).
Consequently, unless $\mathrm{P}=\mathrm{NP}$, there is no XP (and hence no FPT) algorithm parameterized solely by $c$ for this range of $p$.

\section{Conclusion}
In this work, we study the compatibility between EFX and generalized-mean objectives in the setting with few surplus items where EFX is guaranteed to exist.
Our complexity characterization reveal a sharp transition at $p=0$, establishing a dichotomy between the $p \leq 0$ and $p \in (0,1]$ case, both for the compatibility and optimization variants of the problem.
We further quantify the efficiency loss of enforcing EFX using the ``price of EFX'' framework, providing upper and lower bounds.
Finally, we demonstrate that adding even the weakest, objective-free efficiency requirement (Pareto-optimality) makes the problem intractable, while additional structure such as nonzero marginal utilities restores tractability.

Overall, our results delineate when EFX is computationally expensive versus when it is essentially aligned with the welfare objectives, and they motivate approximation and parameterized approaches beyond $c \leq 3$.

\bibliographystyle{plainnat}
\bibliography{abb,bib}

\begin{thebibliography}{44}
\providecommand{\natexlab}[1]{#1}
\providecommand{\url}[1]{\texttt{#1}}
\expandafter\ifx\csname urlstyle\endcsname\relax
  \providecommand{\doi}[1]{doi: #1}\else
  \providecommand{\doi}{doi: \begingroup \urlstyle{rm}\Url}\fi

\bibitem[Amanatidis et~al.(2020)Amanatidis, Markakis, and Ntokos]{amanatidis2020multiple}
Georgios Amanatidis, Evangelos Markakis, and Apostolos Ntokos.
\newblock Multiple birds with one stone: Beating 1/2 for {EFX} and {GMMS} via envy cycle elimination.
\newblock \emph{Theoretical Computer Science}, 841:\penalty0 94--109, 2020.

\bibitem[Amanatidis et~al.(2021)Amanatidis, Birmpas, Filos-Ratsikas, Hollender, and Voudouris]{amanatidis2021mnwefx}
Georgios Amanatidis, Georgios Birmpas, Aris Filos-Ratsikas, Alexandros Hollender, and Alexandros~A. Voudouris.
\newblock Maximum {N}ash welfare and other stories about {EFX}.
\newblock \emph{Theoretical Computer Science}, 863:\penalty0 69--85, 2021.

\bibitem[Amanatidis et~al.(2023)Amanatidis, Aziz, Birmpas, Filos-Ratsikas, Li, Moulin, Voudouris, and Wu]{amanatidis2023fairdivisionprogress}
Georgios Amanatidis, Haris Aziz, Georgios Birmpas, Aris Filos-Ratsikas, Bo~Li, Hervé Moulin, Alexandros~A. Voudouris, and Xiaowei Wu.
\newblock Fair division of indivisible goods: Recent progress and open questions.
\newblock \emph{Artificial Intelligence}, 322:\penalty0 103965, 2023.

\bibitem[Aziz et~al.(2023)Aziz, Huang, Mattei, and Segal-Halevi]{AZIZ2023773}
Haris Aziz, Xin Huang, Nicholas Mattei, and Erel Segal-Halevi.
\newblock Computing welfare-maximizing fair allocations of indivisible goods.
\newblock \emph{European Journal of Operational Research}, 307\penalty0 (2):\penalty0 773--784, 2023.
\newblock ISSN 0377-2217.

\bibitem[Barman and Sundaram(2021)]{barman2021uniformwelfare}
Siddharth Barman and Ranjani~G Sundaram.
\newblock Uniform welfare guarantees under identical subadditive valuations.
\newblock In \emph{Proceedings of the 29th International Joint Conference on Artificial Intelligence (IJCAI)}, pages 46--52, 2021.

\bibitem[Barman and Verma(2022)]{barman2022truthfulmatroid}
Siddharth Barman and Paritosh Verma.
\newblock Truthful and fair mechanisms for matroid-rank valuations.
\newblock In \emph{Proceedings of the 36th AAAI Conference on Artificial Intelligence (AAAI)}, pages 4801--4808, 2022.

\bibitem[Barman et~al.(2018)Barman, Krishnamurthy, and Vaish]{barman2018ef1po}
Siddharth Barman, Sanath~Kumar Krishnamurthy, and Rohit Vaish.
\newblock Finding fair and efficient allocations.
\newblock In \emph{Proceedings of the 19th ACM Conference on Economics and Computation (EC)}, pages 557--574, 2018.

\bibitem[Barman et~al.(2020)Barman, Bhaskar, Krishna, and Sundaram]{barman2020pmeans}
Siddharth Barman, Umang Bhaskar, Anand Krishna, and Ranjani~G Sundaram.
\newblock Tight approximation algorithms for p-mean welfare under subadditive valuations.
\newblock In \emph{Proceedings of the 28th Annual European Symposium on Algorithms (ESA)}, page~11, 2020.

\bibitem[Bei et~al.(2021)Bei, Lu, Manurangsi, and Suksompong]{bei2021price}
Xiaohui Bei, Xinhang Lu, Pasin Manurangsi, and Warut Suksompong.
\newblock The price of fairness for indivisible goods.
\newblock \emph{Theory of Computing Systems}, 65:\penalty0 1069--1093, 2021.

\bibitem[Bertsimas et~al.(2011{\natexlab{a}})Bertsimas, Farias, and Trichakis]{BertsimasEtAl2011}
Dimitris Bertsimas, Vivek~F. Farias, and Nikolaos Trichakis.
\newblock The price of fairness.
\newblock \emph{Operations Research}, 59\penalty0 (1):\penalty0 17--31, 2011{\natexlab{a}}.

\bibitem[Bertsimas et~al.(2011{\natexlab{b}})Bertsimas, Farias, and Trichakis]{bertsimas2011priceoffairness}
Dimitris Bertsimas, Vivek~F. Farias, and Nikolaos Trichakis.
\newblock The price of fairness.
\newblock \emph{Operations Research}, 59\penalty0 (1):\penalty0 17--31, 2011{\natexlab{b}}.

\bibitem[Bouveret and Lemaître(2016)]{bouveret2016conflict}
Sylvain Bouveret and Michel Lemaître.
\newblock Characterizing conflicts in fair division of indivisible goods using a scale of criteria.
\newblock \emph{Autonomous Agents and Multi-Agent Systems}, 30\penalty0 (2):\penalty0 259--290, 2016.

\bibitem[Brams and Taylor(1996)]{BramsTa96}
Steven~J. Brams and Alan~D. Taylor.
\newblock \emph{Fair Division: From Cake-Cutting to Dispute Resolution}.
\newblock Cambridge University Press, 1996.

\bibitem[Bu et~al.(2023)Bu, Song, and Yu]{bu2023efx}
Xiaolin Bu, Jiaxin Song, and Ziqi Yu.
\newblock {EFX} allocations exist for binary valuations.
\newblock In \emph{International Workshop on Frontiers in Algorithmics}, pages 252--262. Springer, 2023.

\bibitem[Bu et~al.(2025)Bu, Li, Liu, Song, and Tao]{bu2025approximability}
Xiaolin Bu, Zihao Li, Shengxin Liu, Jiaxin Song, and Biaoshuai Tao.
\newblock Approximability landscape of welfare maximization within fair allocations.
\newblock In \emph{Proceedings of the 26th ACM Conference on Economics and Computation (EC)}, pages 412--440, 2025.

\bibitem[Budish(2011)]{Budish11}
Eric Budish.
\newblock The combinatorial assignment problem: Approximate competitive equilibrium from equal incomes.
\newblock \emph{Journal of Political Economy}, 119\penalty0 (6):\penalty0 1061--1103, 2011.

\bibitem[Budish and Cantillon(2012)]{budish2012courseallocation}
Eric Budish and Estelle Cantillon.
\newblock The multi-unit assignment problem: Theory and evidence from course allocation at harvard.
\newblock \emph{The American Economic Review}, 102\penalty0 (5):\penalty0 2237--2271, 2012.

\bibitem[Caragiannis et~al.(2012)Caragiannis, Kaklamanis, Kanellopoulos, and Kyropoulou]{caragiannis2012efficiencyoffairdivision}
Ioannis Caragiannis, Christos Kaklamanis, Panagiotis Kanellopoulos, and Maria Kyropoulou.
\newblock The efficiency of fair division.
\newblock \emph{Theory of Computing Systems}, 50:\penalty0 589--610, 2012.

\bibitem[Caragiannis et~al.(2019)Caragiannis, Kurokawa, Moulin, Procaccia, Shah, and Wang]{caragiannis2019unreasonable}
Ioannis Caragiannis, David Kurokawa, Herv{\'e} Moulin, Ariel~D Procaccia, Nisarg Shah, and Junxing Wang.
\newblock The unreasonable fairness of maximum nash welfare.
\newblock \emph{ACM Transactions on Economics and Computation}, 7\penalty0 (3):\penalty0 1--32, 2019.

\bibitem[Celine et~al.(2023)Celine, Dzulfikar, and Koswara]{celine2023egalitarianPrice}
Karen~Frilya Celine, Muhammad~Ayaz Dzulfikar, and Ivan~Adrian Koswara.
\newblock Egalitarian price of fairness for indivisible goods.
\newblock In \emph{Proceedings of the 20th Pacific Rim International Conference on Artificial Intelligence (PRICAI)}, pages 23--28, 2023.

\bibitem[Chaudhury et~al.(2021{\natexlab{a}})Chaudhury, Garg, and Mehta]{chaudhury2021fairandefficient}
Bhaskar~Ray Chaudhury, Jugal Garg, and Ruta Mehta.
\newblock Fair and efficient allocations under subadditive valuations.
\newblock In \emph{Proceedings of the 35th AAAI Conference on Artificial Intelligence (AAAI)}, pages 5269--5276, 2021{\natexlab{a}}.

\bibitem[Chaudhury et~al.(2021{\natexlab{b}})Chaudhury, Kavitha, Mehlhorn, and Sgouritsa]{chaudhury2021little}
Bhaskar~Ray Chaudhury, Telikepalli Kavitha, Kurt Mehlhorn, and Alkmini Sgouritsa.
\newblock A little charity guarantees almost envy-freeness.
\newblock \emph{SIAM Journal on Computing}, 50\penalty0 (4):\penalty0 1336--1358, 2021{\natexlab{b}}.

\bibitem[Chaudhury et~al.(2024)Chaudhury, Garg, and Mehlhorn]{chaudhury_EFX_3agents}
Bhaskar~Ray Chaudhury, Jugal Garg, and Kurt Mehlhorn.
\newblock {EFX} exists for three agents.
\newblock \emph{Journal of the ACM}, 71\penalty0 (1):\penalty0 1--27, 2024.

\bibitem[de~Keijzer et~al.(2009)de~Keijzer, Bouveret, Klos, and Zhang]{dekeijzer2009efpo_sigmahard}
Bart de~Keijzer, Sylvain Bouveret, Tomas Klos, and Yingqian Zhang.
\newblock On the complexity of efficiency and envy-freeness in fair division of indivisible goods with additive preferences.
\newblock In \emph{Proceedings of the 1st International Conference on Algorithmic Decision Theory (ADT)}, pages 98--110, 2009.

\bibitem[Eckart et~al.(2024)Eckart, Psomas, and Verma]{eckart2024pmeans}
Owen Eckart, Alexandros Psomas, and Paritosh Verma.
\newblock On the fairness of normalized p-means for allocating goods and chores.
\newblock In \emph{Proceedings of the 25th ACM Conference on Economics and Computation (EC)}, page 1267, 2024.

\bibitem[Feige et~al.(2021)Feige, Sapir, and Tauber]{feige2021tight}
Uriel Feige, Ariel Sapir, and Laliv Tauber.
\newblock A tight negative example for {MMS} fair allocations.
\newblock In \emph{International Conference on Web and Internet Economics}, pages 355--372. Springer, 2021.

\bibitem[Garg and Murhekar(2023)]{GargMurhekar2023FewValues}
Jugal Garg and Aniket Murhekar.
\newblock Computing fair and efficient allocations with few utility values.
\newblock \emph{Theoretical Computer Science}, 962:\penalty0 113932, 2023.

\bibitem[Garg et~al.(2022)Garg, Husic, and Murhekar]{garg2022tractableMNW}
Jugal Garg, Edin Husic, and Aniket Murhekar.
\newblock Tractable fragments of the maximum nash welfare problem.
\newblock In \emph{Proceedings of the 18th International Conference on Web and Internet Economics (WINE)}, page 362, 2022.

\bibitem[Halpern et~al.(2020)Halpern, Procaccia, Psomas, and Shah]{halpern2020binary}
Daniel Halpern, Ariel~D. Procaccia, Alexandros Psomas, and Nisarg Shah.
\newblock Fair division with binary valuations: One rule to rule them all.
\newblock In \emph{Proceedings of the 16th International Conference on Web and Internet Economics (WINE)}, pages 370--383, 2020.

\bibitem[Hsu(2024)]{hsu2024existence}
Kevin Hsu.
\newblock Existence of {MMS} allocations with mixed manna.
\newblock \emph{arXiv preprint arXiv:2401.07490}, 2024.

\bibitem[Kurokawa et~al.(2016)Kurokawa, Procaccia, and Wang]{kurokawa2016can}
David Kurokawa, Ariel Procaccia, and Junxing Wang.
\newblock When can the maximin share guarantee be guaranteed?
\newblock In \emph{Proceedings of the 30th AAAI Conference on Artificial Intelligence (AAAI)}, pages 523--529, 2016.

\bibitem[Kurz(2014)]{kurz2014pricesmallnumber}
Sascha Kurz.
\newblock The price of fairness for a small number of indivisible items.
\newblock \emph{Operations Research Proceedings}, pages 335--340, 2014.

\bibitem[Kyropoulou et~al.(2020)Kyropoulou, Suksompong, and Voudouris]{KYROPOULOU2020110}
Maria Kyropoulou, Warut Suksompong, and Alexandros~A. Voudouris.
\newblock Almost envy-freeness in group resource allocation.
\newblock \emph{Theoretical Computer Science}, 841:\penalty0 110--123, 2020.

\bibitem[Li et~al.(2024)Li, Liu, Lu, Tao, and Tao]{li2024completelandscapeEF}
Zihao Li, Shengxin Liu, Xinhang Lu, Biaoshuai Tao, and Yichen Tao.
\newblock A complete landscape for the price of envy-freeness.
\newblock In \emph{Proceedings of the 23rd International Conference on Autonomous Agents and Multiagent Systems (AAMAS)}, pages 1183--1191, 2024.

\bibitem[Lipton et~al.(2004)Lipton, Markakis, Mossel, and Saberi]{lipton2004approximately}
Richard~J Lipton, Evangelos Markakis, Elchanan Mossel, and Amin Saberi.
\newblock On approximately fair allocations of indivisible goods.
\newblock In \emph{Proceedings of the 5th ACM Conference on Electronic Commerce (EC)}, pages 125--131, 2004.

\bibitem[Mahara(2023)]{mahara2023efx_twoadditive}
Ryoga Mahara.
\newblock Existence of efx for two additive valuations.
\newblock \emph{Discrete Applied Mathematics}, 340:\penalty0 115--122, 2023.

\bibitem[Mahara(2024)]{mahara2023extension}
Ryoga Mahara.
\newblock Extension of additive valuations to general valuations on the existence of {EFX}.
\newblock \emph{Mathematics of Operations Research}, 49\penalty0 (2):\penalty0 1263--1277, 2024.

\bibitem[Moulin(2019)]{Moulin19}
Hervé Moulin.
\newblock Fair division in the internet age.
\newblock \emph{Mathematical Social Sciences}, 11:\penalty0 407--441, 2019.

\bibitem[Neoh and Teh(2025)]{NeohTe25}
Tzeh~Yuan Neoh and Nicholas Teh.
\newblock Understanding {EFX} allocations: Counting and variants.
\newblock In \emph{Proceedings of the 39th AAAI Conference on Artificial Intelligence (AAAI)}, pages 14036--14044, 2025.

\bibitem[Plaut and Roughgarden(2020)]{plaut2020almost}
Benjamin Plaut and Tim Roughgarden.
\newblock Almost envy-freeness with general valuations.
\newblock \emph{SIAM Journal on Discrete Mathematics}, 34\penalty0 (2):\penalty0 1039--1068, 2020.

\bibitem[Prakash et~al.(2024)Prakash, Ghosal, Nimbhorkar, and Varma]{prakash_EFX_3types}
Vishwa~HV Prakash, Pratik Ghosal, Prajakta Nimbhorkar, and Nithin Varma.
\newblock {EFX} exists for three types of agents.
\newblock \emph{arXiv preprint arXiv:2410.13580}, 2024.

\bibitem[Pratt and Zeckhauser(1990)]{pratt1990inheritance}
John~Winsor Pratt and Richard~Jay Zeckhauser.
\newblock The fair and efficient division of the winsor family silver.
\newblock \emph{Management Science}, 36\penalty0 (11):\penalty0 1293--1301, 1990.

\bibitem[Procaccia(2020)]{procaccia2020enigmatic}
Ariel Procaccia.
\newblock Technical perspective: {A}n answer to fair division's most enigmatic question.
\newblock \emph{Communications of the ACM}, 63\penalty0 (4):\penalty0 118, 2020.

\bibitem[Viswanathan and Zick(2023)]{viswanathan2023generalframework}
Vignesh Viswanathan and Yair Zick.
\newblock A general framework for fair allocation under matroid rank valuations.
\newblock In \emph{Proceedings of the 24th ACM Conference on Economics and Computation (EC)}, pages 1129--1152, 2023.

\end{thebibliography}

\clearpage

\appendix
\begin{center}
\Large
\textbf{Appendix}
\end{center}

\vspace{2mm}
    
\section{Omitted Proofs in Section~\ref{sec:efx+welfare}}
\subsection{Proof of Theorem~\ref{thm:efx_global_p>0_hard}}
    Fix $p\in(0,1]$. Since $x\mapsto x^{1/p}$ is strictly increasing on $\mathbb R_{\ge 0}$, maximizing
    \begin{equation*}
        W_p(A)=\Bigl(\frac{1}{n}\sum_{i\in N} v_i(A_i)^p\Bigr)^{1/p}
    \end{equation*}
    is equivalent to maximizing
    \begin{equation*}
        \Phi_p(A):=\sum_{i\in N} v_i(A_i)^p.
    \end{equation*}
    We split our analysis into two cases, depending on the value of $c$.
    
    \medskip
    \noindent\textbf{Case 1: $c\in\{0,1,2,3\}$.}
    We reduce from the NP-hard \textsc{Partition} problem.
    Let $\{a_1,\dots,a_k\}$ be a \textsc{Partition} instance with total sum $T:=\sum_{j=1}^k a_j$;
    assume $T$ is even (otherwise multiply all $a_j$ by $2$).
    
    Let $n:=k+3$ and define the goods set
    \[
    G=\{g_1,\dots,g_k\}\cup\{x,y\}\cup\{z_0,z_1,\dots,z_c\},
    \]
    so $m=k+2+(c+1)=k+3+c=n+c$.
    
    Choose an integer constant $\Lambda\ge 2$ (depending only on the fixed $p$) such that
    \begin{equation}\label{eq:Lambda_cond}
    (\Lambda T)^p-(\Lambda T/2)^p>(T/2)^p,
    \end{equation}
    equivalently $\Lambda^p(2^p-1)>1$.
    
    Define additive valuations as follows (all unspecified values are $0$):
    \begin{itemize}
    \item For each $j\in[k]$, $v_1(g_j)=v_2(g_j)=a_j$, and $v_i(g_j)=0$ for all $i\notin\{1,2\}$.
    \item $v_1(x)=v_1(y)=v_2(x)=v_2(y)=T/2$, and $v_3(x)=v_3(y)=\Lambda T/2$.
    \item For each $t\in\{0,1,\dots,c\}$, $v_4(z_t)=\Lambda T$ and $v_i(z_t)=0$ for all $i\neq 4$.
    \item Agents $5,6,\dots,k+3$ value every good at $0$.
    \end{itemize}

    \begin{claim}\label{clm:structure}
        Let $\mathcal{A}$ maximize $\Phi_p$ over all allocations. Then:
    (i) $z_0,\dots,z_c\in A_4$;
    (ii) $x,y\in A_3$;
    (iii) each $g_j$ is allocated to agent $1$ or $2$.
    \end{claim}
    
    \begin{proof}
    (i) and (iii) are immediate because if a good $g$ is allocated to an agent $\ell$ with $v_\ell(g)=0$ and there exists an agent $i$ with $v_i(g)>0$, then reassigning $g$ from agent~$\ell$ to agent~$i$ strictly increases $\Phi_p$ for any $p>0$.

    Thus, we focus on showing (ii), i.e., that $x\in A_3$ (the argument for $y$ is identical).
    Suppose $x\notin A_3$, and let $\ell\neq 3$ be such that $x\in A_\ell$.
    Let $\mathcal{A}'$ be the allocation obtained from $\mathcal{A}$ by moving $x$ from agent~$\ell$'s bundle to agent~$3$,
    Then
    \begin{align*}
        \Phi_p(\mathcal{A}')-\Phi_p(\mathcal{A}) & =(v_3(A_3)+v_3(x))^p-v_3(A_3)^p  + (v_\ell(A_\ell)-v_\ell(x))^p-v_\ell(A_\ell)^p.
    \end{align*}
    Here $v_3(x)=\Lambda T/2$ and, since agent $3$ only values $\{x,y\}$ and $x\notin A_3$, we have $v_3(A_3)\le \Lambda T/2$.
    For $0<p\le 1$, $t\mapsto t^p$ is concave, hence the increment
    $h(t):=(t+\delta)^p-t^p$ is nonincreasing in $t$ for each $\delta>0$; thus
    \[
    (v_3(A_3)+\Lambda T/2)^p-v_3(A_3)^p \ \ge\ (\Lambda T)^p-(\Lambda T/2)^p.
    \]
    Also, for $0<p\le 1$, subadditivity $(a+b)^p\le a^p+b^p$ implies
    $(v_\ell(A_\ell)-v)^p-v_\ell(A_\ell)^p\ge -v^p$ for $v=v_\ell(x)$ (since $v_\ell(A_\ell)\ge v$).
    Because $v_\ell(x)\in\{0,T/2\}$,
    \[
    (v_\ell(A_\ell)-v_\ell(x))^p-v_\ell(A_\ell)^p \ \ge\ - (T/2)^p.
    \]
    Combining,
    \[
    \Phi_p(A')-\Phi_p(A)\ \ge\ (\Lambda T)^p-(\Lambda T/2)^p-(T/2)^p\ >\ 0
    \]
    by~\eqref{eq:Lambda_cond}, contradicting optimality. Hence $x\in A_3$, and similarly $y\in A_3$.
    \end{proof}

    Now consider any allocation satisfying Claim~\ref{clm:structure}.
    Let
    \begin{equation*}
        U_1:=\sum_{j:g_j\in A_1} a_j,\qquad U_2:=\sum_{j:g_j\in A_2} a_j,
    \end{equation*}
    so $U_1+U_2=T$. 
    Then agent $3$ has utility $\Lambda T$ (from $\{x,y\}$), agent $4$ has utility $(c+1)\Lambda T$ (from $\{z_0,\dots,z_c\}$), and all other agents contribute $0$ to $\Phi_p$.
    Thus,
    \begin{equation*}
        \Phi_p(A)=U_1^p+U_2^p+(\Lambda T)^p+\bigl((c+1)\Lambda T\bigr)^p,
    \end{equation*}
    so maximizing $\Phi_p$ reduces to maximizing $U_1^p+U_2^p$ under $U_1+U_2=T$.
    If $p\in(0,1)$, strict concavity of $t^p$ implies the unique maximizer is $U_1=U_2=T/2$; if $p=1$, every split is $\Phi_1$-maximizing.

    If the \textsc{Partition} instance is a yes-instance, pick $S\subseteq[k]$ with $\sum_{j\in S}a_j=T/2$ and allocate $\{g_j:j\in S\}$ to agent $1$, $\{g_j:j\notin S\}$ to agent $2$, $\{x,y\}$ to agent $3$, and $\{z_0,\dots,z_c\}$ to agent $4$.
    This allocation is $\Phi_p$-maximizing by the discussion above, and it is EFX: agents $\ge 5$ value everything at $0$; agent $4$'s goods are valued at $0$ by others; agent $3$ values only $\{x,y\}$ and holds both; agents $1,2$ each have value $T/2$ and, after removing either $x$ or $y$ from agent $3$'s bundle, value the remainder at exactly $T/2$; and the EFX relationship between agents $1$ and $2$ hold because removing any one partition good reduces the other's value by at least $0$.

    Conversely, suppose there exists an allocation that is both EFX and $\Phi_p$-maximizing.
    By Claim~\ref{clm:structure}, $A_3=\{x,y\}$ and all $g_j$ are split between agents $1$ and $2$, so $U_1+U_2=T$. 
    EFX for agent $1$ toward agent $3$ after removing $x$ gives
    \begin{equation*}
        U_1=v_1(A_1)\ge v_1(A_3\setminus\{x\})=v_1(\{y\})=T/2,
    \end{equation*}
    and similarly $U_2\ge T/2$. Hence $U_1=U_2=T/2$, which gives us a perfect partition.
    Thus the constructed instance is a yes-instance iff the \textsc{Partition} instance is a yes-instance, proving NP-hardness
    for each $c\in\{0,1,2,3\}$.

\medskip
\noindent\textbf{Case 2: $c<0$.}
Let $d:=-c\in\mathbb N$. We reduce from the $c=0$ setting (already NP-hard by Case~1).
Given any instance $\mathcal I_0=(N_0,G,\mathbf v)$ with $|G|=|N_0|$, form $\mathcal I=(N,G,\mathbf v')$
by adding $d$ dummy agents $D$ with identically-zero valuations and keeping the goods unchanged:
$N:=N_0\cup D$, $G:=G$, $v'_i:=v_i$ for $i\in N_0$, and $v'_d(\cdot)\equiv 0$ for all $d\in D$.
Then $|G|=|N_0|=|N|+c$.

Because every good has positive value for some original agent (as assumed throughout the paper),
no $\Phi_p$-maximizing allocation in $\mathcal I$ assigns any good to a dummy agent:
moving such a good to an original agent who values it positively strictly increases $\Phi_p$ (since $p>0$).
Therefore the set of $\Phi_p$-maximizers in $\mathcal I$ (restricted to original agents) coincides with that in $\mathcal I_0$.
Moreover, adding dummy agents does not affect the EFX property among original agents, and all EFX inequalities involving dummy agents hold trivially because dummy valuations are $0$.
Hence $\mathcal{I}$ admits an EFX $\Phi_p$-maximizer iff $\mathcal I_0$ does, so the $c<0$ setting is also NP-hard.

Combining the two cases completes the proof.
    
\subsection{Proof of Theorem~\ref{prop:pleq0_polytime}}
Since valuations are additive, $v_i(\varnothing)=0$ for all $i$.
If an allocation gives every agent at most one good, then for any $i,j\in N$ either $A_j=\varnothing$
(in which case the EFX$_0$ condition is vacuous) or $A_j=\{g\}$ is a singleton, and then
$A_j\setminus\{g\}=\varnothing$, so
\[
v_i(A_i) \ge 0 = v_i(\varnothing) = v_i(A_j\setminus\{g\}).
\]
Hence any allocation with $|A_i|\le 1$ for all $i$ is EFX$_0$.
We split the remainder of this proof into two cases. 

\medskip\noindent
\textbf{Case 1: $m<n$.}
In every allocation, by the pigeonhole principle, some agent $i$ receives $A_i=\varnothing$, and thus $v_i(A_i)=0$. 
Since $p\leq 0$, this implies $W_p(\mathcal{A})=0$ for every allocation $\mathcal{A}$, so the global optimum equals $0$.
Therefore any allocation that assigns each good to a distinct agent (and leaves the remaining agents empty) is EFX$_0$, and maximizes $W_p$ globally. 
Such an allocation is trivial to construct in $O(m)$ time.

\medskip\noindent
\textbf{Case 2: $m=n$.}
Build the bipartite graph $H=(N,G;E)$ where $(i,g)\in E$ iff $v_i(g)>0$.

\smallskip\noindent
\emph{Subcase 2(a): $H$ has no perfect matching.}
Suppose for contradiction that there exists a complete allocation $\mathcal{A}$ with $v_i(A_i)>0$ for all $i$.
Then each agent must receive at least one good (otherwise her utility would be $0$).
Since $m=n$, this forces every agent to receive exactly one good. Let $\pi:N\to G$ map each
agent to her unique good, so $A_i=\{\pi(i)\}$ for all $i$.
Because $v_i(A_i)=v_i(\pi(i))>0$ for all $i$, every edge $(i,\pi(i))$ lies in $E$, and hence $\pi$ is a
perfect matching of $H$, contradiction. Therefore, in every allocation $\mathcal{A}$, there exists
$i$ with $v_i(A_i)=0$, so (by the convention) $W_p(A)=0$ for all allocations and the global
optimum is $0$. 
Any singleton allocation (i.e., any bijection between agents and goods)
is EFX$_0$ (as shown earlier) and globally optimal.

\smallskip\noindent
\emph{Subcase 2(b): $H$ has a perfect matching.}
Then there exists an allocation $\mathcal{A}$ in which every agent receives one good of strictly
positive value, so $W_p(\mathcal{A})>0$ and therefore the global optimum
$\OPT:=\max\{W_p(\mathcal{B}): \mathcal{B} \in \Pi(G)\}$ satisfies $\OPT>0$.
Let $\mathcal{A}^*$ be a complete allocation with $W_p(\mathcal{A}^*)=\OPT$. If $v_i(A_i^*)=0$ for some
$i$, then $W_p(\mathcal{A}^*)=0$ by convention, contradicting the fact that $\OPT>0$.
Thus, $v_i(A_i^*)>0$ for all $i$, so each agent receives at least one good.
Together with $m=n$, this implies
each agent receives exactly one good. Thus $\mathcal{A}^*$ corresponds to a perfect matching $\pi:N\to G$ in $H$.

It remains to find a perfect matching $\pi$ maximizing $W_p$. 
Since every globally optimal
allocation corresponds to a perfect matching in $H$, it suffices to optimize over perfect matchings of $H$.

\begin{itemize}
\item If $p=0$ (i.e., Nash welfare), then for a matching $\pi$ we have
\begin{equation*}
    W_0(\pi) = \Big(\prod_{i\in N} v_i(\pi(i))\Big)^{1/n}.
\end{equation*}
Since $x\mapsto x^{1/n}$ and $\log(\cdot)$ are strictly increasing on $(0,\infty)$, maximizing $W_0$
over perfect matchings is equivalent to maximizing
$\sum_{i\in N} \log v_i(\pi(i))$ over perfect matchings in $H$.

\item If $p<0$, then for a matching $\pi$ we have
\begin{equation*}
    W_p(\pi) = \Big(\tfrac{1}{n}\sum_{i\in N} v_i(\pi(i))^p\Big)^{1/p}.
\end{equation*}
Since $1/p<0$, the map $x\mapsto x^{1/p}$ is strictly decreasing on $(0,\infty)$, so maximizing $W_p$
is equivalent to minimizing $\sum_{i\in N} v_i(\pi(i))^p$ over perfect matchings in $H$.
\end{itemize}

In both cases above ($p = 0$ and $p < 0$), we obtain a standard assignment problem (maximum-weight or minimum-cost perfect
matching) on the bipartite graph $H$, which can be solved in polynomial time (e.g., by min-cost flow algorithms; when one prefers the Hungarian method, one may complete the graph by adding sufficiently bad dummy edges). 
Output the corresponding singleton allocation; it is EFX$_0$ and globally $W_p$-optimal by construction.

\medskip\noindent
Note that we can assume valuations and $p$ to be rational, and thus the objectives use weights $\log v_i(g)$ (for $p=0$) or $v_i(g)^p$ (for $p<0$) that are computable
reals, and computing them to sufficiently many bits of precision (polynomial in the input size) suffices to recover an exact optimal matching.

\subsection{Explanation of Example~\ref{example:compatibility}}
Observe that maximizing $W_0$ is equivalent to maximizing the Nash product $\prod_{i\in N} v_i(A_i)$, and under our convention any maximizer must satisfy $v_i(A_i)>0$ for all $i\in N$.
First, $v_2(A_2)>0$ forces $g_4\in A_2$, since agent~2 values only $g_4$ positively.
Next, suppose (toward contradiction) that $g_1\notin A_1$. Then $v_1(A_1)>0$ forces $g_2\in A_1$ and hence
$v_1(A_1)\le v_1(g_2)=1$, while agent~3 can obtain utility at most
$v_3(A_3)\le v_3(\{g_1,g_3\})=3$. Therefore, $\prod_{i\in N} v_i(A_i)\le 1\cdot 5\cdot 3=15$.
On the other hand, $\mathcal{A}^*$ gives utilities $(v_1(A_1^*),v_2(A_2^*),v_3(A_3^*))=(6,5,1)$ and Nash product
$6\cdot 5\cdot 1=30>15$. Hence every Nash-optimal allocation must have $g_1\in A_1$.

Given $g_1\in A_1$ and $g_4\in A_2$, the remaining goods are $g_2$ and $g_3$. Allocating both to agent~3 gives Nash product $5\cdot 5\cdot (1+\tfrac{1}{10})=\tfrac{55}{2}<30$, and allocating $g_2$ to agent~3 and $g_3$ to agent~1 gives Nash product $5\cdot 5\cdot \tfrac{1}{10}=\tfrac{5}{2}<30$. Thus $\mathcal{A}^*$ is the unique
Nash-optimal allocation.

Finally, $\mathcal{A}^*$ violates EFX: for $i=3$, $j=1$, and $g=g_2\in A^*_1$ we have $v_3(g)>0$ but $v_3(A^*_3)=v_3(\{g_3\})=1 < v_3(A^*_1\setminus\{g_2\})=v_3(\{g_1\})=2$.

\subsection{Proof of Theorem~\ref{thm:decide-find-efx-global-pmean-maximizer}}
    We present the argument for EFX$_0$, and detail at the end of the proof how it can be adapted to work for EFX.
    
    Define $\OPT:=\max_{\mathcal{A}\in \Pi(G)} W_p(\mathcal{A})$ and $\OPT_{\mathrm{EFX}_0}:=\max_{\mathcal{A}\in \Pi_{\mathrm{EFX}_0}(G)} W_p(\mathcal{A})$.
    
    We give a polynomial-time procedure that outputs an allocation $\mathcal{A}'\in \Pi_{\mathrm{EFX}_0}(G)$
    with $W_p(\mathcal{A}')=\OPT$ if and only if $\OPT_{\mathrm{EFX}_0}=\OPT$, and otherwise
    outputs `no'. The procedure has three steps.
    
    \smallskip
    \noindent\textbf{Step 1: Compute $\OPT$ and an optimizer $\mathcal{A}^{\mathrm{opt}}\in\Pi(G)$.}
    Construct the bipartite graph $B=(N,G;E)$ where $(i,g)\in E$ iff $v_i(g)>0$, and compute a maximum
    matching $M$ in $B$.
    If $|M|<n$, then no allocation can give every agent strictly positive utility: indeed, if some allocation
    $\mathcal{A}$ satisfies $v_i(A_i)>0$ for all $i$, then for each $i$ we can select a good $g_i\in A_i$ with
    $v_i(g_i)>0$, and the edges $\{(i,g_i)\}_{i\in N}$ would form a matching of size $n$, a contradiction.
    Thus, every allocation $\mathcal{A}$ has $v_i(A_i)=0$ for some $i$, and we have
    $W_p(\mathcal{A})=0$ for all $\mathcal{A}$. 
    Therefore, $\OPT=0$; fix any allocation and denote it by $\mathcal{A}^{\mathrm{opt}}$, and proceed to Step~2.
    
    Assume henceforth that $|M|=n$. 
    Then, there exists an allocation with $v_i(A_i)>0$ for all $i$, so $\OPT>0$. 
    Let $\mathcal{A}^*$ attain $\OPT$. 
    Necessarily, $v_i(A^*_i)>0$ for all $i \in N$ (otherwise $W_p(\mathcal{A}^*)=0<\OPT$), and each agent receives at least one good.
    
    Let $c:=m-n$. For any allocation $\mathcal{A}$ with $|A_i|\ge 1$ for all $i \in N$, define
    \begin{equation*}
        H(\mathcal{A}):=\{ i\in N : |A_i|\ge 2\} \quad \text{and} \quad S(\mathcal{A}):=\bigcup_{i\in H(\mathcal{A})} A_i.
    \end{equation*}
    Since $\sum_{i\in N}(|A_i|-1)=m-n=c$ and agents not in $H(\mathcal{A})$ have $|A_i|=1$, it follows that
    \begin{equation*}
        \sum_{i\in H(\mathcal{A})} (|A_i|-1)=c,
    \end{equation*}
    so $|H(\mathcal{A})|\le c$. 
    Moreover, the disjointness of bundles gives us
    \begin{align*}
        |S(\mathcal{A})|=\sum_{i\in H(\mathcal{A})}|A_i|
    =\sum_{i\in H(\mathcal{A})}\bigl(1+(|A_i|-1)\bigr) =|H(\mathcal{A})|+c \le 2c.
    \end{align*}
    In particular, $|S(\mathcal{A}^*)|\leq 2c$.
    
    We now enumerate all possibilities for the non-singleton bundles in $\mathcal{A}^*$.
    For each $\kappa\in\{1,\dots,c\}$, enumerate all triples $(H,S,(B_i)_{i\in H})$ such that:
    (i) $H\subseteq N$ with $|H|=\kappa$;
    (ii) $S\subseteq G$ with $|S|=\kappa+c$;
    (iii) $(B_i)_{i\in H}$ is a labeled partition of $S$ with $|B_i|\ge 2$ for all $i\in H$.
    Fix such a triple and let $L:=N\setminus H$ and $R:=G\setminus S$. Then $|L|=|R|$.
    
    Discard the current triple if $v_i(B_i)=0$ for some $i\in H$.
    Also discard it if there is no bijection $\pi:L\to R$ with $v_\ell(\pi(\ell))>0$ for all $\ell\in L$
    (equivalently, if the bipartite graph on $L\times R$ containing edges $(\ell,g)$ with $v_\ell(g)>0$
    has no perfect matching).
    
    For a fixed surviving triple, optimizing $W_p$ over feasible completions reduces to an assignment problem
    on $L\times R$:
    \begin{itemize}
        \item If $p=0$, then (under positive utilities) maximizing $W_0$ is equivalent to maximizing
        \[
            \sum_{\ell\in L}\log v_\ell(\pi(\ell))
        \]
        over feasible bijections $\pi:L\to R$.
        \item If $p<0$, then (under positive utilities) maximizing $W_p$ is equivalent to minimizing
        \[
            \sum_{\ell\in L} v_\ell(\pi(\ell))^p
        \]
        over feasible bijections $\pi:L\to R$.
    \end{itemize}
    Solve the resulting assignment instance, form the corresponding allocation, and keep the best one across
    all iterations; denote the best found allocation by $\mathcal{A}^{\mathrm{opt}}$.
    
    We now prove the correctness of Step 1.
    The triple $\bigl(H(\mathcal{A}^*),S(\mathcal{A}^*),(A^*_i)_{i\in H(\mathcal{A}^*)}\bigr)$ appears in the enumeration.
    In that iteration, the bijection induced by $\mathcal{A}^*$ is feasible. The assignment subroutine returns an optimal
    completion for this fixed triple, hence produces an allocation with $W_p$ value at least $W_p(\mathcal{A}^*)=\OPT$.
    Since $\OPT$ is the global maximum, we obtain $W_p(\mathcal{A}^{\mathrm{opt}})=\OPT$.
    
    \smallskip
    \noindent\textbf{Step 2: Compute $\mathrm{OPT_{EFX_0}}$ and an optimizer $\mathcal{A}^{\mathrm{opt_0}}\in \Pi_{\mathrm{EFX_0}}(G)$.}
    We compute $\mathrm{OPT_{EFX_0}}$ via the same enumeration procedure as in Step~1, but we restrict
    to completions that satisfy EFX$_0$.
    
    For each enumerated triple $(H,S,(B_i)_{i\in H})$, let $L:=N\setminus H$ and $R:=G\setminus S$
    (so $|L|=|R|$). For each $i\in H$ and agent $x\in N$, define
    \[
    \tau_x(B_i):=\max_{g\in B_i} v_x(B_i\setminus\{g\}),
    \]
    and for each light agent $\ell\in L$ define
    \[
    \tau_\ell := \max_{i\in H} \tau_\ell(B_i),
    \]
    with the convention $\tau_\ell=0$ if $H=\varnothing$.
    
    \medskip
    \noindent\textbf{Claim.}
    Consider an allocation $\widetilde{\mathcal{A}}$ that gives each heavy agent $i\in H$ the bundle $B_i$ and gives each
    light agent $\ell\in L$ a singleton $\{\pi(\ell)\}$ for some bijection $\pi:L\to R$.
    Then $\widetilde{\mathcal{A}}$ is EFX$_0$ if and only if:
    \begin{enumerate}
    \item[(i)] for all distinct $a,i\in H$, we have $v_a(B_a)\ge \tau_a(B_i)$; and
    \item[(ii)] for every $\ell\in L$, we have $v_\ell(\pi(\ell))\ge \tau_\ell$.
    \end{enumerate}
    
    \noindent\emph{Proof of the claim.}
    Fix such $\widetilde{\mathcal{A}}$. The EFX0 constraints are
    $v_x(\widetilde{A}_x)\ge v_x(\widetilde{A}_y\setminus\{g\})$ for all $x,y\in N$ and all $g\in \widetilde{A}_y$.
    If $y\in L$ then $\widetilde{A}_y$ is a singleton and $\widetilde{A}_y\setminus\{g\}=\varnothing$, so these inequalities
    are automatic. Thus only $y\in H$ matters, i.e., $y=i$ with bundle $B_i$.
    Fix $i\in H$ and $x\in N$. The constraints from $x$ toward $i$ are
    $v_x(\widetilde{A}_x)\ge v_x(B_i\setminus\{g\})$ for all $g\in B_i$, equivalently
    $v_x(\widetilde{A}_x)\ge \max_{g\in B_i} v_x(B_i\setminus\{g\})=\tau_x(B_i)$.
    If $x=a\in H$ then $\widetilde{A}_x=B_a$, giving condition (i). If $x=\ell\in L$ then $\widetilde{A}_x=\{\pi(\ell)\}$,
    so we require $v_\ell(\pi(\ell))\ge \tau_\ell(B_i)$ for all $i\in H$, i.e.,
    $v_\ell(\pi(\ell))\ge \max_{i\in H}\tau_\ell(B_i)=\tau_\ell$, which is (ii).
    Conversely, (i) and (ii) imply all the EFX$_0$ inequalities by the same case analysis. \hfill$\square$
    
    \medskip
    Now, discard the current triple if $v_i(B_i)=0$ for some $i\in H$ (since for $p\le 0$ this forces
    welfare $0$), or if condition (i) of the claim fails.
    
    Next, build a bipartite graph on $L\times R$ that contains an edge $(\ell,g)$ if and only if
    $v_\ell(g)>0$ and $v_\ell(g)\ge \tau_\ell$.
    If there is no perfect matching, discard this iteration.
    
    For a fixed surviving triple, optimizing $W_p$ over feasible bijections $\pi:L\to R$ reduces to an
    assignment problem:
    \begin{itemize}
    \item If $p=0$, maximizing $W_0$ is equivalent (under positive utilities) to maximizing
    $\sum_{\ell\in L}\log v_\ell(\pi(\ell))$ over feasible perfect matchings.
    \item If $p<0$, maximizing $W_p$ is equivalent (under positive utilities) to minimizing
    $\sum_{\ell\in L} v_\ell(\pi(\ell))^p$ over feasible perfect matchings.
    \end{itemize}
    Solve the resulting assignment instance, form the corresponding allocation $\widetilde{\mathcal{A}}$, and keep the
    best one across all iterations; denote the best found EFX$_0$ allocation by $\mathcal{A}^{\mathrm{opt_0}}$ and its value by
    $\mathrm{OPT_{EFX_0}}$.
    
    If no iteration gives us a feasible completion with all agents receiving positive utility, then
    $\mathrm{OPT_{EFX_0}}=0$; in this case set $\mathcal{A}^{\mathrm{opt_0}}$ to be any EFX$_0$ allocation (which exists for $m\leq n+3$).

    We show the correctness of Step 2.
    Let $\widehat{\mathcal{A}}$ be an EFX$_0$ allocation achieving $\OPT_{\mathrm{EFX}_0}$. 
    Since $\OPT_{\mathrm{EFX}_0}>0$ implies every agent has positive utility, $\widehat{\mathcal{A}}$ has at most $c$ heavy agents and $|S(\widehat{\mathcal{A}})|\le 2c$. 
    Thus, its heavy part appears in the enumeration. 
    In that iteration, the light agent assignment induced by $\widehat{\mathcal{A}}$ satisfies conditions (i)–(ii), hence corresponds to a feasible perfect matching in the constructed graph. 
    The assignment solver returns an optimal completion for that iteration, with welfare at least $W_p(\widehat{\mathcal{A}})$. Taking the best over all iterations gives us welfare $\OPT_{\mathrm{EFX}_0}$.
    
    \smallskip
    \noindent\textbf{Step 3.}
    Since $\Pi_{\mathrm{EFX}_0}(G)\subseteq \Pi(G)$, we have $\OPT_{\mathrm{EFX}_0}\le \OPT$.
    If $\OPT_{\mathrm{EFX}_0}=\OPT$, output $\mathcal{A}^{\mathrm{opt}_0}$; it is $\mathrm{EFX}_0$ and globally $W_p$-optimal.
    Otherwise $\OPT_{\mathrm{EFX}_0}<\OPT$, and then no $\mathrm{EFX}_0$ allocation attains $\OPT$, so output `no'.
    The case $\OPT=0$ is covered as well, since then necessarily $\OPT_{\mathrm{EFX}_0}=0$.
    
    Finally, we prove the running time of the described procedure.
    For fixed $c$, Step~1 is polynomial: $|S|\le 2c$ gives $O(m^{2c})$ choices for $S$ and $O(n^c)$ choices for $H$,
    and the number of labeled partitions of a set of size at most $2c$ into at most $c$ parts of size at least $2$ is a constant
    (depending only on $c$). 
    Each iteration solves one assignment problem in polynomial time.
    Step~2 has the same enumeration size as Step~1 and solves one assignment instance per iteration (after filtering infeasible triples and edges), hence it is also polynomial for fixed $c$. The weight computations (logs for $p=0$, or $p$-powers for fixed $p<0$) are handled exactly as in Step~1.
    In Step~1, the assignment instances use edge weights $\log v_\ell(g)$ (when $p=0$) or $v_\ell(g)^p$ (when $p<0$);
    these are computable reals for fixed $p$ and rational input valuations, which we can assume without loss of generality.
    
    The proof for EFX is identical, except that whenever we define $\tau_x(B_i)=\max_{g\in B_i} v_x(B_i\setminus\{g\})$,
    we instead use $\tau_x^*(B_i):=\max_{g\in B_i:\ v_x(g)>0} v_x(B_i\setminus\{g\})$, with $\tau_x^*(B_i):=0$ if $\{g\in B_i: v_x(g)>0\}=\varnothing$,
    and we require $v_\ell(g)\ge \tau_\ell^*$ in the matching stage.

\subsection{Proof of Theorem~\ref{thm:efx0-efx-welfare-opt}}
    We present the proof/algorithm for EFX$_0$; at the end, we then mention how it can be adapted to work for EFX as well. 

    Let $c:=m-n$. Since $m \le n+3$, we have $c\in\{1,2,3\}$ and $m=n+c$.
    Call agents with $|A_i|\ge 2$ \emph{heavy}, and the other agents \emph{light}.
    The \emph{heavy part} of an allocation $\mathcal{A}$ is then the set of heavy agents together with the goods they hold (i.e., the non-singleton bundles).
    We now compute a $W_p$-optimal EFX$_0$ allocation by enumerating the
    heavy part and solving one assignment problem per enumeration.

    First, assume the nondegenerate case
    \begin{equation*}
        \OPT:=\max_{A\in\Pi_{\mathrm{EFX}_0}(G)} W_p(A) > 0.
    \end{equation*}
    Then every optimal allocation $\mathcal{A}^*$ satisfies $v_i(A^*_i)>0$ for all $i$, and therefore
    each agent receives at least one good.
    
    Let $\mathcal{A}$ be any allocation in which every agent receives at least one good.
    Define the set of heavy agents
    \begin{equation*}
        H(\mathcal{A}):=\{i:|A_i|\ge 2\}
    \end{equation*}
    and let $k:=|H(A)|$.
    Since $m=n+c$ and each of the $n$ agents receives at least one good, the number of
    extra goods beyond one per agent is exactly $c$, hence
    \begin{equation*}
        \sum_{i\in H(A)}\bigl(|A_i|-1\bigr)=c.
    \end{equation*}
    In particular, $k\le c$. Moreover, the total number of goods held by the heavy agents equals
    \begin{align*}
        \left|\bigcup_{i\in H(A)}A_i\right|
     =\sum_{i\in H(A)}|A_i| & =\sum_{i\in H(A)}\bigl(1+(|A_i|-1)\bigr) \\
    & 
    =k+c \le 2c \le 6.
    \end{align*}
    Thus, in any optimal $\mathcal{A}^*$ with $\OPT>0$, at most $c$ agents have non-singleton bundles, and the union of all non-singleton bundles contains at most $6$ goods.
    
    We enumerate all possibilities for the heavy part as follows.
    For each integer $k\in\{0,1,\ldots,c\}$:
    \begin{itemize}
    \item If $k=0$ and $c>0$, skip this $k$ (because with $m=n+c>n$ at least one bundle must be non-singleton).
    \item Enumerate all subsets $H\subseteq N$ with $|H|=k$.
    \item Enumerate all subsets $S\subseteq G$ with $|S|=k+c$.
    \item Enumerate all surjective assignments $\beta:S\to H$ such that each
    $B_i:=\beta^{-1}(i)$ has size $|B_i|\ge 2$ (equivalently, a labeled partition of $S$
    into $k$ parts of size at least $2$).
    \end{itemize}
    Fix one such triple $(H,S,(B_i)_{i\in H})$ from the enumeration.
    Let $L:=N\setminus H$ (light agents) and $R:=G\setminus S$ (remaining goods).
    Because $|G|=n+c$ and $|S|=k+c$, we have $|R|=n-k=|L|$.
    
    For each $i\in H$ and each agent $x\in N$, define
    \[
    \tau_x(B_i):=\max_{g\in B_i} v_x(B_i\setminus\{g\}).
    \]
    For each light agent $\ell\in L$, define the threshold
    \[
    \tau_\ell:=\max_{i\in H}\tau_\ell(B_i),
    \]
    with the convention $\tau_\ell=0$ if $H=\varnothing$.
    
    \begin{lemma} \label{lem:efx0-threshold-characterization}
    Consider an allocation $\widehat{\mathcal{A}}$ that gives each heavy agent $i\in H$ the bundle $B_i$
    and gives each light agent $\ell\in L$ a singleton $\{\pi(\ell)\}$ for some bijection $\pi:L\to R$.
    Then $\widehat{\mathcal{A}}$ is \emph{EFX$_0$} if and only if:
    \begin{enumerate}
    \item[(i)] For all distinct $a,i\in H$, $v_a(B_a)\ge \tau_a(B_i)$.
    \item[(ii)] For every $\ell\in L$, $v_\ell(\pi(\ell))\ge \tau_\ell$.
    \end{enumerate}
    \end{lemma}
    
    \begin{proof}
        Fix such $\widehat{\mathcal{A}}$. Then, suppose $\widehat{\mathcal{A}}$ is EFX$_0$: for all $x,y\in N$ and all $g\in \widehat{A}_y$,
        \[
        v_x(\widehat{A}_x)\ \ge\ v_x(\widehat{A}_y\setminus\{g\}).
        \]
        If $y\in L$, then $\widehat{A}_y$ is a singleton; removing its unique good gives us $\varnothing$,
        whose value is $0$ for every agent, so these inequalities are automatically satisfied.
        Thus, the only nontrivial inequalities are those with $y\in H$.
        
        Fix $i\in H$ and $x\in N$. The inequalities from $x$ toward $i$ are
        $v_x(\widehat{A}_x)\ge v_x(B_i\setminus\{g\})$ for all $g\in B_i$, i.e.,
        \[
        v_x(\widehat{A}_x)\ \ge\ \max_{g\in B_i} v_x(B_i\setminus\{g\})\ =\ \tau_x(B_i).
        \]
        The $x=i$ case is trivial since $v_a(B_a)\ge v_a(B_a\setminus\{g\})$ for all $g\in B_a$.
        If $x\in H$, then $\widehat{A}_x=B_x$ and this is exactly (i).
        If $x\in L$, then $\widehat{A}_x=\{\pi(x)\}$ and this becomes
        $v_x(\pi(x))\ge \tau_x(B_i)$ for all $i\in H$, i.e.,
        $v_x(\pi(x))\ge \max_{i\in H}\tau_x(B_i)=\tau_x$, which is (ii).
        Conversely, (i) and (ii) imply all EFX$_0$ inequalities by the same reasoning.
    \end{proof}
    
    Then, for the current enumerated $(H,S,(B_i)_{i\in H})$, first check condition (i) of the lemma.
    If it fails, discard this iteration.
    Now impose $\OPT>0$. 
    Since $p\le 0$ and $W_p=0$ whenever some agent gets utility $0$, we only need to consider completions where every agent has strictly positive utility.
    Thus, discard this iteration if $v_i(B_i)=0$ for some $i\in H$.
    
    Define a bipartite graph with left side $L$ and right side $R$; include an edge $(\ell,g)$ iff
    \[
    v_\ell(g)>0
    \quad\text{and}\quad
    v_\ell(g)\ge \tau_\ell.
    \]
    Any perfect matching $\pi:L\to R$ in this graph gives us a completion $\widehat{\mathcal{A}}$ that is EFX$_0$
    by the lemma and has all-light utilities positive.
    
    Among such feasible matchings, optimizing $W_p$ reduces to a standard assignment objective:
    \begin{itemize}
    \item If $p=0$ (Nash welfare), maximizing $W_0(\widehat{\mathcal{A}})$ is equivalent to maximizing
    $\sum_{i\in N}\log v_i(\widehat{A}_i)$. Since the heavy bundles are fixed, this is equivalent to maximizing
    $\sum_{\ell\in L}\log v_\ell(\pi(\ell))$ over perfect matchings (a maximum-weight perfect matching problem).
    \item If $p<0$, then all utilities are positive and
    \[
    W_p(\widehat{\mathcal{A}})
    =
    \left(\frac{1}{n}\sum_{i\in N} v_i(\widehat{A}_i)^p\right)^{1/p}.
    \]
    Because $1/p<0$, maximizing $W_p(\widehat{\mathcal{A}})$ is equivalent to minimizing $\sum_{i\in N} v_i(\widehat{A}_i)^p$.
    Again, the heavy part is constant in this iteration, so it suffices to minimize
    $\sum_{\ell\in L} v_\ell(\pi(\ell))^p$ over perfect matchings (a minimum-cost perfect matching problem).
    \end{itemize}
    As $p$ is fixed, edge weights/costs can be computed in polynomial time to the needed precision; for integer $p$
    they are rational.
    If no perfect matching exists, discard this iteration; otherwise compute an optimal perfect matching
    (e.g., Hungarian algorithm in $O(n^3)$) and record the resulting allocation $\widehat{\mathcal{A}}$.
    Finally, among all recorded allocations across all iterations, return one with maximum $W_p$.

    If the above enumeration results in no feasible completion (i.e., no iteration admits a perfect matching satisfying $v_\ell(g)>0$ and $v_\ell(g)\ge \tau_\ell$), then no EFX$_0$ allocation gives strictly positive utility to every agent. 
    This implies $\max_{\mathcal{A}\in\Pi_{\mathrm{EFX}_0}(G)} W_p(\mathcal{A})=0$; in this case we output any EFX$_0$ allocation (guaranteed to exist for $m\leq n+3$).
    
    Now, we prove correctness.
    Let $\mathcal{A}^*$ be an optimal EFX$_0$ allocation with $\OPT>0$.
    By Step~1, $H(\mathcal{A}^*)$ has size $k^*\le c$ and the union
    $S^*:=\bigcup_{i\in H(\mathcal{A}^*)}A^*_i$ has size $k^*+c\le 6$.
    Therefore the enumeration includes the iteration with
    $H=H(\mathcal{A}^*)$, $S=S^*$, and $B_i=A^*_i$ for all $i\in H(\mathcal{A}^*)$.
    In that iteration, $\mathcal{A}^*$ corresponds to a feasible perfect matching on $L\times R$
    (and satisfies the lemma's EFX$_0$ constraints by definition).
    The matching solver returns a completion whose objective (sum of logs for $p=0$, or sum of $p$-powers for $p<0$)
    is at least as good as $\mathcal{A}^*$'s, thus giving us welfare at least $W_p(\mathcal{A}^*)=\OPT$.
    Since $\OPT$ is the maximum, the algorithm returns an allocation attaining $\OPT$.
    
    Finally, we prove the polynomial runtime of the algorithm.
    Since $c\le 3$ is constant, we have that $k+c\le 6$.
    The number of choices of $S$ is $\binom{m}{k+c} = \mathcal{O}(m^{k+c}) = \mathcal{O}(m^6)$, and the number of labeled partitions of a set with size at most $6$ into $k\le 3$ parts of size $\ge 2$ is a constant.
    Thus the enumeration is polynomial, and each iteration solves one assignment instance in $\mathcal{O}(n^3)$ time.
    Hence the total running time is polynomial in $n$ and $m$.

    Now, we mention how the proof can be adapted to work for EFX.
    Note that the fact EFX$_0$ implies EFX alone does not justify that an EFX$_0$-optimal allocation is also EFX-optimal, because the feasible region is strictly larger for EFX.
    However, the same framework optimizes over EFX after a single change: whenever we define a threshold of the form
    \begin{equation*}
        \tau_x(B_i)=\max_{g\in B_i} v_x(B_i\setminus\{g\}),
    \end{equation*}
    replace it by
    \begin{equation*}
        \tau_x^*(B_i)
    := \max_{\substack{g\in B_i\\ v_x(g)>0}} v_x(B_i\setminus\{g\}),
    \end{equation*}
    with $\tau_x^*(B_i):=0$ if $\{g\in B_i:v_x(g)>0\}=\varnothing$.
    Then define $\tau_\ell^*:=\max_{i\in H}\tau_\ell^*(B_i)$, and in the matching stage require $v_\ell(g)\ge \tau_\ell^*$.
    Lemma~\ref{lem:efx0-threshold-characterization} and the remainder of the proof carry over, because EFX constraints are exactly the same as EFX$_0$ constraints except that they quantify only over goods with positive value to the envying agent.

\subsection{Proof of Theorem~\ref{thm:efx-pmean-opt-hard}}
    We provide the proof for EFX, and remark at the end why the same reduction holds for EFX$_0$ as well.

    Fix $p\in(0,1]$. Define
    \begin{equation*}
        \Phi_p(\mathcal{A}) := \sum_{i\in N} v_i(A_i)^p.
    \end{equation*}
    Since
    \begin{equation*}
        W_p(\mathcal{A})=\left(\frac{1}{n}\sum_{i\in N} v_i(A_i)^p\right)^{1/p} =\left(\frac{1}{n}\Phi_p(\mathcal{A})\right)^{1/p},
    \end{equation*}
    and $x\mapsto x^{1/p}$ is strictly increasing for $p>0$, maximizing $W_p$
    is equivalent to maximizing $\Phi_p$.
    
    We first show NP-hardness for the case where $c \in \{0,1,2,3\}$.
    We reduce from the NP-hard problem \textsc{Partition}. 
    Let $a_1,\dots,a_k$ be positive integers with total sum $T:=\sum_{j=1}^k a_j$, and assume $T$ is even (otherwise replace each $a_j$ by $2a_j$). Fix any constant $c\in\{0,1,2,3\}$.

    Let $n_0:=k+3$ and define agents
    \begin{equation*}
        N:=\{1,2,\dots,n_0\}\cup D \quad \text{and} \quad D:=\{d_1,\dots,d_c\}
    \end{equation*}
    (with $D=\varnothing$ when $c=0$), so $|N|=n_0+c$.
    Let $H:=\{x_1,y_1,\dots,x_c,y_c\}$ and define goods
    \begin{equation*}
        G:=\{g_1,\dots,g_k,x,y,z\}\cup H,
    \end{equation*}
    so
    \begin{equation*}
        m=|G|=k+3+2c=(n_0+c)+c=|N|+c.
    \end{equation*}
    Choose an integer constant $\Lambda \geq 2$ (depending only on fixed $p$) such that
    \begin{equation}\label{eq:Lambda}
        \Lambda^p(1-2^{-p})>2^{-p}.
    \end{equation}
    This is possible since $2^p-1>0$ for $p>0$, and \eqref{eq:Lambda} is equivalent to $\Lambda^p(2^p-1)>1$.

    Define agents' valuations over the goods as follows:
\begin{itemize}
    \item \textit{Partition goods.} For each $j\in[k]$, $v_1(g_j)=v_2(g_j)=a_j$ and $v_i(g_j)=0\ \ \forall i\in N\setminus\{1,2\}$.
    \item \textit{Special goods $x,y$.} $v_1(x)=v_1(y)=v_2(x)=v_2(y)=\tfrac{T}{2}$, $v_3(x)=v_3(y)=\tfrac{\Lambda T}{2}$, and $v_i(x)=v_i(y)=0$ for all $i\in N\setminus\{1,2,3\}$.
    \item \textit{Special good $z$.} For each $i\in\{4,\dots,n_0\}$ set $v_i(z)=\Lambda T$,
    and for $i\in\{1,2,3\}\cup D$ set $v_i(z)=0$.
    \item \textit{Dummy goods (only if $c\geq 1$).} For each $t\in[c]$, $v_{d_t}(x_t)=v_{d_t}(y_t)=1$, $v_{d_t}(g)=0  \text{ for all } g\in G\setminus\{x_t,y_t\}$, and for every original agent $i\in\{1,\dots,n_0\}$ and every $h\in H$ set $v_i(h)=0$.
\end{itemize}
All numbers have $\mathcal{O}(\log T)$ bits (since $\Lambda$ is a constant for fixed $p$), so the
construction is polynomial-time. Moreover, $\Pi_{\mathrm{EFX}}(G)\neq\varnothing$ in the
constructed instance (e.g., assign each good to a distinct agent except that each $d_t$
receives $\{x_t,y_t\}$; then every nonempty bundle is a singleton or a pair that all other
agents value at $0$, so EFX holds).

\medskip
\noindent\textbf{Claim.}
Let $\mathcal{A}$ be an allocation maximizing $\Phi_p$ over all allocations.
Then:
\begin{enumerate}
\item $x,y\in A_3$;
\item $z\in A_i$ for some $i\in\{4,\dots,n_0\}$;
\item for each $t\in[c]$, $\{x_t,y_t\}\subseteq A_{d_t}$;
\item each partition good $g_j$ is allocated to agent $1$ or agent $2$.
\end{enumerate}

\noindent\emph{Proof of the claim.}
(2)--(4) follow immediately because $\Phi_p$ is strictly increasing in each agent's
utility (as $p>0$) and each listed good has positive value only for the specified agent(s).
It remains to show (1), i.e., $x\in A_3$ (the argument for $y$ is identical).

Assume for contradiction that $x\notin A_3$, and let $\ell\neq 3$ be such that $x\in A_\ell$.
Let $\mathcal{A}'$ be obtained from $\mathcal{A}$ by moving good $x$ from agent~$\ell$ to agent~$3$.
Then
\begin{align*}
    \Phi_p(\mathcal{A}')-\Phi_p(\mathcal{A}) & = (v_3(A_3)+v_3(x))^p-v_3(A_3)^p \\
    & \quad \ + (v_\ell(A_\ell)-v_\ell(x))^p-v_\ell(A_\ell)^p.
\end{align*}
Here $v_3(x)=\Lambda T/2$, and because agent $3$ values only $x$ and $y$ and $x\notin A_3$,
we have $v_3(A_3)\le \Lambda T/2$. For $0<p\le 1$, the function $t\mapsto t^p$ is concave, hence
$t\mapsto (t+\delta)^p-t^p$ is nonincreasing in $t$ for each fixed $\delta>0$, so
\[
(v_3(A_3)+\Lambda T/2)^p-v_3(A_3)^p \;\ge\; (\Lambda T)^p-(\Lambda T/2)^p.
\]
Also $v_\ell(x)\in\{0,T/2\}$ and $v_\ell(A_\ell)\ge v_\ell(x)$. Using subadditivity
$(a+b)^p\le a^p+b^p$ for $0<p\le 1$ (apply it to $v_\ell(A_\ell)=(v_\ell(A_\ell)-v_\ell(x))+v_\ell(x)$), we get
\[
(v_\ell(A_\ell)-v_\ell(x))^p-v_\ell(A_\ell)^p \geq -\,v_\ell(x)^p \;\ge\; -(T/2)^p.
\]
Therefore,
\begin{align*}
    \Phi_p(\mathcal{A}')-\Phi_p(\mathcal{A}) & \geq (\Lambda T)^p-(\Lambda T/2)^p-(T/2)^p \\ 
    & = T^p\big(\Lambda^p(1-2^{-p})-2^{-p}\big)>0
\end{align*}
by \eqref{eq:Lambda}, contradicting optimality of $\mathcal{A}$. Hence $x\in A_3$, and similarly $y\in A_3$.
\qed

\medskip
For any allocation $\mathcal{A}$ satisfying the claim, define
\begin{equation*}
    U_1:=\sum_{j:\,g_j\in A_1} a_j \quad \text{and} \quad U_2:=\sum_{j:\,g_j\in A_2} a_j,
\end{equation*}
so $U_1+U_2=T$. Then agent $3$ has utility $\Lambda T$ from $\{x,y\}$, exactly one
agent in $\{4,\dots,n_0\}$ has utility $\Lambda T$ from $z$, each $d_t$ has utility $2$ from
$\{x_t,y_t\}$, and all remaining agents have utility $0$. Hence
\begin{equation}\label{eq:Phi-form}
\Phi_p(A)=c\cdot 2^p + 2(\Lambda T)^p + U_1^p+U_2^p.
\end{equation}
Since $t\mapsto t^p$ is concave for $0<p\le 1$, Jensen's inequality gives us
$U_1^p+U_2^p\le 2(T/2)^p$ under $U_1+U_2=T$. 
Define $B_c := c\cdot 2^p + 2(\Lambda T)^p + 2(T/2)^p$.
Let $\mathcal{A}^{\max}$ be a global maximizer of $\Phi_p$ over $\Pi(G)$. 
Now, by the claim, $\mathcal{A}^{\max}$ satisfies (1)--(4), and therefore $\Phi_p(\mathcal{A}^{\max})\le B_c$, i.e., $\max_{\mathcal{B}\in\Pi(G)}\Phi_p(\mathcal{B})\le B_c$.
Moreover, if $p\in(0,1)$, strict concavity implies equality iff $U_1=U_2=T/2$. If $p=1$, then $t\mapsto t$ is linear, so equality holds for all $U_1,U_2$ satisfying $U_1+U_2=T$ (for any $c$).

Now, suppose the \textsc{Partition} instance is a yes-instance, so there is $S_1\subseteq[k]$ with $\sum_{j\in S_1} a_j=T/2$ and $S_2=[k]\setminus S_1$. 
Define the allocation $\mathcal{A}^*$ as follows: $A^*_1=\{g_j:j\in S_1\}$, $A^*_2=\{g_j:j\in S_2\}$, $A^*_3=\{x,y\}$, $A^*_4=\{z\}$, $A^*_{d_t}=\{x_t,y_t\} \forall t\in[c]$, and give $\varnothing$ to every other original agent in $\{5,\dots,n_0\}$.
Then $U_1=U_2=T/2$, so \eqref{eq:Phi-form} gives $\Phi_p(\mathcal{A}^*)=B_c$.

We prove that $\mathcal{A}^*$ is EFX.
Dummy agents only value their own pair, so any comparison involving a dummy bundle is trivial.
Agent $3$ values only $x$ and $y$ and holds both, so she does not envy any other bundle.
Agents $1$ and $2$ each have utility $T/2$, and for each $g\in\{x,y\}$, $v_1(A^*_3\setminus\{g\})=T/2\le v_1(A^*_1)$ and $v_2(A^*_3\setminus\{g\})=T/2\le v_2(A^*_2)$,
so EFX holds with respect to agent $3$. Between agents $1$ and $2$, for any $g\in A^*_2$,
\[
v_1(A^*_2\setminus\{g\})=v_1(A^*_2)-v_1(g)\le T/2=v_1(A^*_1),
\]
and symmetrically for agent $2$ with respect to agent $1$. Finally $A^*_4\setminus\{z\}=\varnothing$,
so all remaining EFX inequalities involving agent $4$ hold:
for any $i\ge 5$, $A_i=\varnothing$ and $i$ values every bundle minus one good at $0$ (since the only good $i$ values positively is $z$, which is in a singleton bundle). Also, agent $4$ values all goods except $z$ at $0$, so she does not envy anyone.

Let $\mathcal{A}$ be an EFX allocation maximizing $\Phi_p$ over $\Pi_{\mathrm{EFX}}(G)$.
Since $\mathcal{A}^*$ is feasible with value $B_c$, we have $\max_{B\in\Pi_{\mathrm{EFX}}(G)}\Phi_p(B)\ge B_c$.
Also $\Pi_{\mathrm{EFX}}(G)\subseteq \Pi(G)$, so $\max_{B\in\Pi_{\mathrm{EFX}}(G)}\Phi_p(B)\le \max_{B\in\Pi(G)}\Phi_p(B)\le B_c$.
Hence $\Phi_p(\mathcal{A})=B_c$, so in particular $\mathcal{A}$ is also a global $\Phi_p$-maximizer and satisfies the claim. 
Thus $x,y\in A_3$ and all partition goods are split between agents $1$ and $2$.

Applying EFX for agent $1$ with respect to agent $3$ and the good $x\in A_3$ gives us
\begin{equation*}
    U_1=v_1(A_1)\ \ge\ v_1(A_3\setminus\{x\})=v_1(\{y\})=T/2,
\end{equation*}
and similarly $U_2\ge T/2$, so $U_1=U_2=T/2$. 
Thus, given an algorithm that outputs an EFX allocation maximizing $\Phi_p$, we run it on the constructed instance and compute $U_1=\sum_{j:g_j\in A_1} a_j$. Output `yes' iff $U_1=T/2$. By the argument above, this solves \textsc{Partition}.
This shows the optimization problem is NP-hard for each $c\in\{0,1,2,3\}$.

Next, we describe how to extend this result to any fixed negative $c$ (for the case of $m < n$).
Fix $c<0$ and let $t=-c$. Reduce from the $c=0$ construction by adding $t$ dummy agents with identically zero valuations and keeping the goods/valuations of the original agents unchanged. In a yes-instance, the EFX allocation $\mathcal{A}^*$ from the $c=0$ case remains EFX after adding the new agents (give them empty bundles), and it achieves the same value $B_0$. Also, the global upper bound $B_0$ on $\max_{\mathcal{B}\in\Pi(G)}\Phi_p(\mathcal{B})$ remains valid because adding zero agents does not increase $\Phi_p$ for any allocation. Hence the EFX-optimum in the augmented instance equals $B_0$, so any EFX-optimal allocation must also be a global maximizer and the earlier argument forces a perfect partition.

Note that in our reduction, we constructing an instance with special goods $x,y,z$ and (when $c>0$) dummy pairs $(x_t,y_t)$; in a yes-instance it exhibits an allocation $\mathcal{A}^*$ that (i) achieves the global upper bound $B_c$ for $\Phi_p$, and (ii) is EFX. 
In this construction, $\mathcal{A}^*$ is in fact EFX$_0$ as well: whenever an agent values every good in a compared bundle at $0$, removing any one of those goods still leaves the agent's value at $0$, so the extra EFX$_0$ inequalities are automatically satisfied. 
Therefore the EFX$_0$-feasible optimum also attains $B_c$, forcing any EFX$_0$-optimal solution to be a global $\Phi_p$-maximizer with the same forced structural properties. 
Finally, where applicable, we only removed goods from a nonempty bundle when the good is positively valued by agents $1$ and $2$, so the exact same inequalities are required under both EFX and EFX$_0$.
Hence an algorithm that optimizes over $\Pi_{\mathrm{EFX}}$ or over $\Pi_{\mathrm{EFX}_0}$ would decide \textsc{Partition}, proving NP-hardness for both cases.

\section{Omitted Proofs in Section~\ref{sec:price}}
\subsection{Proof of Theorem~\ref{thm:price_>0_lowerbound}}
    We prove the bound by exhibiting a family of instances.
    For $\varepsilon > 0$, let $\mathcal{I}_\varepsilon = (N, G, \mathbf{v}^\varepsilon)$ denote the instance defined below.
    Since $v_i^\varepsilon(g) > 0$ for all $i \in N$ and $g \in G$, we have $\Pi_\mathrm{EFX}(G) = \Pi_{\mathrm{EFX}_0}(G)$ in $\mathcal{I}_\varepsilon$, and thus $\mathrm{PoEFX}_0^{\mathcal{I}_\varepsilon}(p, c) = \mathrm{PoEFX}^{\mathcal{I}_\varepsilon}(p,c)$.

    Consider an instance with $n$ agents and $m=n+c$ goods for some fixed integer $c \leq 3$, and assume $m \geq 1$. 
    Let agents valuations be defined as follows (for simplicity we use $\mathbf{v}$ rather than $\mathbf{v}^\varepsilon$): for every $g \in G$,
    \begin{equation*}
    	v_1(g) = 1 \quad \text{and} \quad v_i(g) = \varepsilon \text{ for all $i \in N \setminus \{1\}$}.
    \end{equation*}
   	
   	Now, note that for any $p\in(0,1]$, the function $f(x) := x^p$ is strictly increasing.
   	
   	Let $\mathcal{B} = (B_1,\dots,B_n)$ be some allocation such that $|B_1| < m$. 
    Fix $\varepsilon > 0$ to be sufficiently small (so that the following inequality holds for every allocation $\mathcal{B}$ with $|B_1| < m$),
   	\begin{equation*}
   		\PM(\mathcal{B})^p = \frac{1}{n} \sum_{i \in N} v_i(B_i)^p < \frac{m^p}{n}
   	\end{equation*}
   	If we consider the allocation $\mathcal{B}' = (B'_1,\dots, B'_n)$ that allocates all $m$ goods to agent~$1$, 
   	\begin{equation*}
   		\PM(\mathcal{B}')^p = \frac{1}{n} \sum_{i \in N} v_i(B'_i)^p = \frac{m^p}{n} + 0 = \frac{m^p}{n} > \PM(\mathcal{B})^p.
   	\end{equation*}
    Thus, for $\mathcal{I}_\varepsilon$, allocating all $m$ goods to agent~$1$ would maximize the $p$-mean welfare; denote
   	\begin{equation} \label{eqn:lowerbound_W_opt_unconstrained}
   		\OPT_p(\mathcal{I}_\varepsilon) := \max_{A\in\Pi(G)} W_p(\mathcal{A}) = W_p(\mathcal{B}') = \frac{m}{n^{1/p}}
   	\end{equation}
    
    Next, we examine the structure of EFX$_0$ allocations for this instance.
    Consider any allocation $\mathcal{A} = (A_1,\dots, A_n)$.
    If $|A_1| = \varnothing$, then $|A_1| = 0 \leq \lfloor \frac{m+n-1}{n}\rfloor$ trivially. Otherwise, fix any $g \in A_1$ and apply EFX$_0$ for agents~$i$ and $1$ with this $g$.
    For any $i \geq 2$,
    \begin{align*}
    	& |A_i| \cdot \varepsilon = v_i(A_i) \geq v_i(A_1) - \varepsilon  =  (|A_1| - 1) \cdot \varepsilon \\
    	\iff\, & |A_i| \geq |A_1| - 1 \\
    	\iff\, & \sum_{i \in N \setminus \{1\}} |A_i| \geq (n-1) \cdot (|A_1|-1) \\
    	\iff\, & m - |A_1| \geq (n-1) \cdot |A_1| - (n-1) \\
    	\iff\, & n \cdot |A_1| \leq m + n - 1,
    \end{align*}
    where the first line is a necessary condition for EFX$_0$. Thus, in order for any allocation $\mathcal{A}$ to be EFX$_0$, we must have that $|A_1| \leq \frac{m+n-1}{n}$, or equivalently, since $|A_1|$ is an integer,
    \begin{equation}
    	|A_1| \leq \left\lfloor \frac{m+n-1}{n} \right\rfloor.
    \end{equation}
    In particular, if $m \leq n$ (i.e., $c \leq 0$), then $\frac{m+n-1}{n} < 2$  and hence $\lfloor \frac{m+n-1}{n} \rfloor = 1$, so every EFX$_0$ allocation $|A_1| \leq 1$.

	Next, we show that there exists an EFX$_0$ allocation $\mathcal{A}^* = (A^*_1,\dots, A^*_n)$ with $|A^*_1| = \left\lfloor \frac{m+n-1}{n} \right\rfloor$.
	Let $\mathcal{A}^*$ be an allocation satisfying the following conditions: for $i \geq 2$,
	\begin{enumerate}[(i)]
		\item $|A_i| \in \{\left\lfloor \frac{m+n-1}{n} \right\rfloor-1, \left\lfloor \frac{m+n-1}{n} \right\rfloor\}$,
		\item exactly $r:= m + n - 1 - n \cdot \left\lfloor \frac{m+n-1}{n} \right\rfloor$ agents have bundle size of exactly $\left\lfloor \frac{m+n-1}{n} \right\rfloor$, and the $n-1-r$ remaining agents have bundle size of exactly $\left\lfloor \frac{m+n-1}{n} \right\rfloor-1$, and
		\item $|A_1| = \left\lfloor \frac{m+n-1}{n} \right\rfloor$.
	\end{enumerate}
	Observe that given the constraints above,
	\begin{align*}
		A_1 + \sum_{i \in N \setminus \{1\}} |A_i| 
		& = \left\lfloor \frac{m+n-1}{n} \right\rfloor + r \cdot \left\lfloor \frac{m+n-1}{n} \right\rfloor + (n-1-r) \cdot \left(\left\lfloor \frac{m+n-1}{n} \right\rfloor - 1 \right) \\
		& = n \cdot \left\lfloor \frac{m+n-1}{n} \right\rfloor - (n-1-r) \\
		& = n \cdot \left\lfloor \frac{m+n-1}{n} \right\rfloor -n + 1 + \left( m + n - 1 - n \cdot \left\lfloor \frac{m+n-1}{n} \right\rfloor \right)\\
		& = m.
	\end{align*}
	Together with the fact the value of each agent for every item is the same, we can be sure that the allocation $\mathcal{A}^*$ is valid and exists.
	
	We now prove that $\mathcal{A}^*$ is EFX$_0$.
	For any $i \geq 2$, conditions (i) and (iii) gives us 
	\begin{equation*}
		|A_i| \geq \left\lfloor \frac{m+n-1}{n} \right\rfloor - 1 = |A_1|-1.
	\end{equation*}
	Thus, we get that
	\begin{align*}
		v_i(A_i) = \varepsilon \cdot |A_i| \geq \varepsilon \cdot (|A_1| - 1) = v_i(A_1 \setminus \{g\})
	\end{align*}
	 for any $g \in A_1$. Thus, any agent $i \in N \setminus \{1\}$ will not envy agent~$1$ after removing any item (in particular, the lowest-valued one) from agent~$1$'s bundle.
	 
	 For any two agents $i,j \geq 2$, $i \neq j$, condition (i) gives us $|A_i|,|A_j| \in \left\{ \left\lfloor \frac{m+n-1}{n} \right\rfloor - 1, \left\lfloor \frac{m+n-1}{n} \right\rfloor \right\}$.
	 Thus, we have that
	 \begin{equation*}
	 	|A_i| \geq |A_j| - 1 \quad \text{and} \quad |A_j| \geq |A_i| - 1.
	 \end{equation*}
	 Since $v_i(A_i) = v_j(A_i) = \varepsilon \cdot |A_i|$, $v_i(A_j) = v_j(A_j) = \varepsilon \cdot |A_j|$, and each agent $i \in N \setminus \{1\}$ values every good at $\varepsilon$, they will not envy each other after dropping any good.
	 
	 Finally, we have that for any $i \geq 2$ and any $g \in A_i$, we have
     \begin{equation*}
         v_1(A_i\setminus\{g\}) = |A_i|-1 \le |A_1| = v_1(A_1).
     \end{equation*}
	  Thus, agent~$1$ will not envy any agent $i$ up to any good.
	 
	Thus, $\mathcal{A}^*$ satisfies EFX$_0$.
	
	Consider any EFX$_0$ allocation $\mathcal{A} = (A_1,\dots, A_n)$. Then, note that $v_1(A_1) = |A_1|$ and $v_i(A_i) = \varepsilon \cdot |A_i|$ for all $i \geq 2$, which vanishes as $\varepsilon \rightarrow 0$.
	The only asymptotically non-negligible term in $\sum_{i \in N} v_i(A_i)^p$ is $|A_1|^p$.
	Thus, by taking $\varepsilon \rightarrow 0$, we get that
	\begin{align*}
		\max_{\mathcal{A}\in\Pi_{\mathrm{EFX}_0}(G)} W_p(\mathcal{A}) & =  \left( \frac{v_i(A_i)^p}{n}  \right)^{1/p} = \frac{|A_1|}{n^{1/p}} + o(1) = \frac{\left\lfloor \frac{m+n-1}{n} \right\rfloor}{n^{1/p}} + o(1)
	\end{align*}
    Combining (\ref{eqn:lowerbound_W_opt_unconstrained}) with the bound $|A_1|\le \left\lfloor\frac{m+n-1}{n}\right\rfloor$ for every EFX$_0$ allocation and the construction of an EFX$_0$ allocation achieving $|A_1|=\left\lfloor\frac{m+n-1}{n}\right\rfloor$, we obtain for $\mathcal{I}_\varepsilon$ that 
    \begin{align*}
        \mathrm{PoEFX}^{\mathcal{I}_\varepsilon}_0(p,c) =\frac{\OPT_p(\mathcal{I}_\varepsilon)}{\max_{\mathcal{A}\in\Pi_{\mathrm{EFX}_0}(G)} W_p(\mathcal{A})} =\frac{m}{\left\lfloor\frac{m+n-1}{n}\right\rfloor}-o(1).
    \end{align*}
    Taking the supremum over $\varepsilon > 0$ (equivalently, $\varepsilon \rightarrow 0^+$) gives us
    \begin{equation*}
        \mathrm{PoEFX}_0(p,c) \geq \frac{m}{\left\lfloor\frac{m+n-1}{n}\right\rfloor}-o(1) = \frac{n+c}{\left\lfloor\frac{2n+c-1}{n}\right\rfloor}-o(1).
    \end{equation*}
    Moreover, because all goods are positively valued  by all agents in $\mathcal{I}_\varepsilon$, EFX and EFX$_0$ coincide, so the same lower bound holds for $\mathrm{PoEFX}(p,c)$.

\subsection{Proof of Theorem~\ref{thm:price_>0_upperbound}}
    We first show the bound for EFX$_0$.
    Before starting with the proof of the upper bound, we first show a useful property.
    
    Let $m\le n+3$ and start from a partial EFX$_0$ allocation where only agent $i$ holds a singleton $\{g\}$ and all other agents hold empty bundles. 
    Then there must exist a consisting of $m-1$ single-item insertions, each followed by a (possibly empty) sequence of EFX$_0$-preserving exchanges that extends this partial allocation to a complete EFX$_0$ allocation in which agent $i$ keeps $g$.
    
    To see this, start from the partial allocation with only $g$ assigned to $i$ and every other agent has an empty bundle; this is trivially EFX$_0$. Add the remaining items one by one. For each insertion, apply \citet{mahara2023extension}'s potential‑based augmentation (for $m\le n+3$) to restore EFX$_0$. Because $\{g\}$ is a singleton and appears at the bottom of every other agent's envy‑comparison (removing its only item gives us the empty bundle), we can maintain an invariant that $g$ is never reallocated: if an exchange would move $g$, swap the orientation of the corresponding cycle (or tie‑break in the potential) to reassign a different item on that step.
    
    Let $v^*:=\max_{i \in N, \, g \in G}v_i(g)$.
    Also let $\OPT_p(\mathcal{I}) := \max_{\mathcal{A}\in \Pi(G)} W_p(\mathcal{A})$, and  $\OPT^{\mathrm{EFX_0}}_p(\mathcal{I}) := \max_{\mathcal{A}\in \Pi_{\mathrm{EFX0}}(G)}  W_p(\mathcal{A})$.
    
    For any allocation $\mathcal{A} = (A_1, \dots, A_n)$ and any $p \in (0, 1]$,
    \begin{align*}
        \PM(\mathcal{A}) \le\ W_1(\mathcal{A}) & = \frac{1}{n} \sum_{i \in N} v_i(A_i) \leq \frac1n\sum_{g\in G}\max_{i \in N} v_i(g) \leq \frac{mv^*}{n}.
    \end{align*}
    In particular, $\OPT_p(\mathcal{I}) \leq mv^*/n$.

    Next, let $(i^*, g^*) := \argmax_{i \in N, g \in G} v_i(g)$.

    If $m \leq n$ (equivalently, $c \leq 0$), construct a (complete) allocation $\widehat{\mathcal{A}}$ by allocating $g^*$ to agent $i^*$, and allocating  each remaining good in $G \setminus \{g^*\}$ to a distinct agent in $N \setminus \{i^*\}$ (possible since $m-1 \leq n-1$); leave all other agents' bundle empty.
    Then every bundle has size at most $1$, so $\widehat{\mathcal{A}}$ is EFX$_0$.
    Thus,
    \begin{equation*}
        W_p(\widehat{\mathcal{A}}) \geq \left( \frac{1}{n} \cdot (v^*)^p \right)^{1/p} = \frac{v^*}{n^{1/p}}.
    \end{equation*}
    Therefore, $\max_{\mathcal{A} \in \Pi_{\mathrm{EFX}_0}} W_p(\mathcal{A}) \geq \frac{v^*}{n^{1/p}}$ also holds when $m < n$.

    By the earlier augmentation property \citep{mahara2023extension}, starting from the partial EFX$_0$ allocation with only $g^*$ allocated to $i^*$, we can extend to a complete EFX$_0$ allocation $\widehat{\mathcal{A}}^*$ while keeping $g^*$ in $i^*$'s bundle.
    Then, we have that
    \begin{align*}
        \PM(\widehat{\mathcal{A}}^*) 
        = \left( \frac{1}{n} v_{i^*}(\widehat{A}^*_{i^*})^p + \sum_{j \in N \setminus \{i\}} v_j(\widehat{A}^*_j)^p \right)^{1/p}  \geq \left( \frac{1}{n} v_{i^*}(\widehat{A}^*_{i^*})^p \right)^{1/p} = \frac{v^*}{n^{1/p}},
    \end{align*}
    since $v_{i^*}(\widehat{A}^*_{i^*}) \geq v_{i^*}(g^*) = v^*$.
    Moreover, since $\max_{\mathcal{A} \in \Pi_{\mathrm{EFX}_0}(G)} \PM(\mathcal{A}) \geq \PM(\widehat{\mathcal{A}}^*)$, we get 
    \begin{align*}
        \mathrm{PoEFX}^\mathcal{I}_0(p,c)  = \frac{\OPT_p(\mathcal{I})}{\OPT^{\mathrm{EFX_0}}_p(\mathcal{I})}  \leq \frac{mv^* / n}{v^*/n^{1/p}} = (n+c) \cdot n^{1/p-1}.
    \end{align*}
    Since the above bound holds for every instance $\mathcal{I}$ with $m=n+c$, taking $\sup_\mathcal{I}$ gives us $\mathrm{PoEFX}_0(p,c) \leq (n+c)\cdot n^{1/p-1}$, as desired.

   Since $\Pi_{\mathrm{EFX_0}}(G)\subseteq \Pi_{\mathrm{EFX}}(G)$, we have $\OPT^{\mathrm{EFX}}_p(\mathcal{I}) = \max_{\mathcal{A} \in \Pi_\mathrm{EFX}(G)} W_p(\mathcal{A}) \geq \OPT^{\mathrm{EFX_0}}_p(\mathcal{I})$,
so the same upper bound holds for $\mathrm{PoEFX}(p,c)$ as well.

\subsection{Proof of Theorem~\ref{thm:price_leq0_lowerbound}}
    We first prove the bound for $\mathrm{PoEFX}_0(p,c)$, and at the end, describe how this can be extended to hold for $\mathrm{PoEFX}(p,c)$ as well.

    We first consider the case where $c < 0$ (i.e., $m < n$).
    Then, in every allocation $\mathcal{A} = (A_1,\dots, A_n)$, there exists an agent $i \in N$ with $A_i = \varnothing$.
    Since $p \leq 0$, this means $W_p(\mathcal{A}) = 0$ for every allocation $\mathcal{A} \in \Pi(G)$.
    Thus, $\max_{\mathcal{A} \in \Pi(G)} W_p(\mathcal{A}) = 0$.
    Since $\Pi_\mathrm{EFX}(G) \subseteq \Pi(G)$, we also have that $\max_{\mathcal{A} \in \Pi_\mathrm{EFX}(G)} W_p(\mathcal{A}) = 0$ (and similarly for EFX$_0$).
    Therefore, $\mathrm{PoEFX}(p,c) = \mathrm{PoEFX}_0(p,c) = 0$.
    We may therefore assume $c \geq 0$, i.e., $c \in \{0,1,2,3\}$ for the remainder of the proof.

    Fix $n\ge 2$ and $c\in\{0,1,2,3\}$, where $m=n+c$. Let the set of goods be $G=\{p_2,\dots,p_n\} \cup S$, where $|S|=c+1$.
    Let agents' valuation functions be defined as follows:
    \begin{itemize}
	\item For agent~$1$: $v_1(g)=\frac{1}{c+1}$ for every $g\in S$, and $v_1(p_i)=0$ for $i\geq 2$.
        \item For each $i\in\{2,\dots,n\}$: $v_i(p_i)=1$ and $v_i(g)=1 + \varepsilon$ for every $g\in S$ and some small $\varepsilon > 0$; all other goods have value $0$ to agent~$i$.
    \end{itemize}
        
    Define the allocation $\mathcal{A^*} = (A_1^*,\dots, A_n^*)$ as follows: allocate all $c+1$ goods in $S$ to agent~$1$ and each good $p_i$ to agent $i\geq 2$.
    Then, under $\mathcal{A}^*$, every agent receives utility $1$.
    Thus, $W_p(\mathcal{A}^*) = 1$ for every $p < 0$, and $W_0(\mathcal{A}^*) = \left( \prod_{i=1}^n 1\right)^{1/n} = 1$.
    Thus, letting $\OPT_p(\mathcal{I}) := \max_{\mathcal{A}\in\Pi(G)} W_p(\mathcal{A})$, we have that $\OPT_p(\mathcal{I}) \geq 1$ for all $p \leq 0$.
    
    Now, let $\mathcal{A} = (A_1, \dots, A_n)$ be any EFX$_0$ allocation. 
    For agents $i\ge 2$, the only private good they value is $p_i$, and it contributes at most $1$ to at most one agent's utility. 
    Thus, the total contribution of the set $\{p_2,\dots,p_n\}$ to $\sum_{i=2}^n v_i(A_i)$ is at most $n-1$, regardless of how those goods are allocated.
    Then, observe that 
    \begin{equation*}
        |A_1 \cap S| \leq 2.
    \end{equation*}
    To see why, suppose that $|A_1 \cap S| \geq 3$. Then, since $c \leq 3$, there are at most $c+1-|A_1 \cap S| \leq 1$ left in $S$ that can be allocated to agents in $N\setminus \{1\}$. Thus, for any $i \in N \setminus \{1\}$ and $g \in A_1$,
    \begin{equation*}
        v_i(A_i) \leq v_i(p_i) + 1+\varepsilon = 2 + \varepsilon < 2 + 2\varepsilon \leq v_i(A_1 \setminus \{g\}),
    \end{equation*}
    contradicting the assumption that $\mathcal{A}$ is an EFX$_0$ allocation.
    Thus, in every EFX$_0$ allocation $\mathcal{A}$, $|A_1 \cap S| \in \{0,1,2\}$. 
    We only need to upper bound $\max_{\mathcal{A}\in\Pi_{\mathrm{EFX}_0}(G)}$ $W_p(\mathcal{A})$. If $\max_{\mathcal{A}\in\Pi_{\mathrm{EFX}_0}(G)}$ $W_p(\mathcal{A})=0$, while $\OPT_p(\mathcal{I}) > 0$, then $\mathrm{PoEFX}_0^\mathcal{I}(p,c) = +\infty$ and the theorem holds trivially.
    Otherwise, fix an EFX$_0$ allocation $\mathcal{A}$ with $W_p(\mathcal{A})>0$, which implies $v_i(A_i) > 0$ for all $i$, and hence $\gamma \geq 1$.
    Thus, $|A_1 \cap S| \in \{1,2\}$.

    Next, we get an upper bound for the maximum $p$-mean welfare for any EFX$_0$ allocation when $p = 0$ (i.e., Nash welfare).
    Let $\gamma = |A_1 \cap S|$ and fix some $\gamma \in \{1,2\}$ where $\gamma \leq c+1$.
    Then, the total number of items from $S$ that is allocated to agents in $N \setminus \{1\}$ is $c+1-\gamma$.

    Each good in $S$ contributes exactly $1+\varepsilon$ to the utility of any agent $i\ge 2$, so the total contribution of the t goods of S that are not given to agent 1 to the sum $\sum_{i=2}^n v_i(A_i)$ is exactly $t(1+\varepsilon)$. Each private good $p_j$ contributes value 1 to at most one agent's utility (only agent j values it), so the total contribution of $\{p_2,\dots,p_n\}$ to $\sum_{i=2}^n v_i(A_i)$ is at most $n-1$. Thus,
    \begin{equation*}
        \sum_{i = 2}^n v_i(A_i) \leq (n-1) + t(1+\varepsilon).
    \end{equation*}
    Then,
    \begin{align*}
        \prod_{i=2}^n v_i(A_i) & \leq \left(\frac{(n-1)+t(1+\varepsilon)}{n-1}\right)^{n-1} =\left(1+\frac{t(1+\varepsilon)}{n-1}\right)^{n-1}.
    \end{align*}
    Also, $v_1(A_1)=\gamma/(c+1)$. Thus, we get that 
    \begin{align*}
        W_0(\mathcal{A}) & = \left(v_1(A_1)\cdot\prod_{i=2}^n u_i\right)^{1/n}  \leq \left(\frac{\gamma}{c+1}\cdot\left(1+\frac{(c+1-\gamma)(1+\varepsilon)}{n-1}\right)^{n-1}\right)^{1/n}.
    \end{align*}
    Then,
    \begin{align*}
        \max_{\mathcal{A} \in \Pi_{\mathrm{EFX}_0}(G)} W_0(\mathcal{A}) 
        & \leq \Big( \frac{\gamma}{c+1} \cdot \Big(1+\tfrac{c+1-\gamma}{n-1}(1+\varepsilon)\Big)^{n-1} \Big)^{1/n} \\
        & \leq \max_{\gamma \in \{1,2\}} \Big( \frac{\gamma}{c+1} \cdot \Big(1+\tfrac{c+1-\gamma}{n-1}(1+\varepsilon)\Big)^{n-1} \Big)^{1/n} \\
         & \underset{\varepsilon \rightarrow 0}{\rightarrow} \max_{\gamma \in \{1,2\}} \Big( \frac{\gamma}{c+1} \cdot \Big(1+\tfrac{c+1-\gamma}{n-1}\Big)^{n-1} \Big)^{1/n},
    \end{align*}
    since $\varepsilon>0$ can be chosen arbitrarily small (and rational), the supremum price over instances is at least the limit as $\varepsilon\to 0^+$.

    Also note that for any $p \leq 0$,
    \begin{equation*}
        \max_{\mathcal{A} \in \Pi_{\mathrm{EFX}_0}(G)}\PM(\mathcal{A}) \leq \max_{\mathcal{A} \in \Pi_{\mathrm{EFX}_0}(G)}W_0(\mathcal{A}).
    \end{equation*}
    Combining this with the above, we have that
    \begin{equation*}
        \max_{\mathcal{A} \in \Pi_{\mathrm{EFX}_0}(G)}\PM(\mathcal{A}) \leq \max_{\gamma \in \{1,2\}} \Big( \frac{\gamma}{c+1} \cdot \Big(1+\tfrac{c+1-\gamma}{n-1}\Big)^{n-1} \Big)^{1/n}.
    \end{equation*}
    Consequently, for the constructed instance $\mathcal{I}$, we have that
    \begin{align*}
        \mathrm{PoEFX}_0^\mathcal{I}(p,c) =\frac{\OPT_p(\mathcal{I})}{\max_{\mathcal{A} \in \Pi_{\mathrm{EFX}_0}(G)} \PM(\mathcal{A})}.        
    \end{align*}
    Since $\OPT_p(\mathcal{I}) \geq 1$, and we have shown that 
    \begin{align*}
         \max_{\mathcal{A}\in\Pi_{\mathrm{EFX0}}(G)} W_p(\mathcal{A})
        & \leq \max_{\gamma\in\{1,2\}} \left( \frac{\gamma}{c+1}\cdot\left(1+\frac{c+1-\gamma}{n-1}\right)^{n-1} \right)^{1/n},
    \end{align*}
    it follows that
    \begin{equation*}
        \mathrm{PoEFX}_0^\mathcal{I}(p,c) \geq \frac{1}{\max_{\gamma \in \{1,2\}} \Big( \frac{\gamma}{c+1} \cdot \Big(1+\tfrac{c+1-\gamma}{n-1}\Big)^{n-1} \Big)^{1/n}}.
    \end{equation*}
    Taking the supremum over all instances with $m=n+c$, we get the bound on $\mathrm{PoEFX}_0(p,c)$ as desired.
    
    Now, we show that the same construction also gives the stated lower bound for $\mathrm{PoEFX}^\mathcal{I}(p,c)$.
    Indeed, fix $p\leq 0$ and consider the instance defined above.
    Let $\mathcal{A}=(A_1,\dots,A_n)$ be an EFX allocation maximizing $W_p$ among EFX allocations.
    If $W_p(\mathcal{A})=0$, then $\mathrm{PoEFX}^\mathcal{I}(p,c)=\frac{\OPT_p(\mathcal{I})}{0}=+\infty$ and the bound is trivial.
    Otherwise $W_p(\mathcal{A}) > 0$, implying $v_i(A_i)>0$ for all $i \in N$, and in particular agent~$1$ must receive at least one good from $S$.
    Let $\gamma:=|A_1\cap S|\ge 1$.
    
    We claim that $\gamma\leq 2$ also holds under EFX.
    Suppose for contradiction that $\gamma\ge 3$.
    Since $|S|=c+1\le 4$, at most one good from $S$ remains outside $A_1$.
    Hence every agent $i\in\{2,\dots,n\}$ receives at most one good from $S$, so
    \[
    v_i(A_i)\ \le\ v_i(p_i) + (1+\varepsilon)\ =\ 2+\varepsilon.
    \]
    On the other hand, for any $g\in A_1\cap S$ we have $|A_1\cap S|\ge 3$ and thus
    $A_1\setminus\{g\}$ contains at least two goods from $S$, giving
    \[
    v_i(A_1\setminus\{g\})\ \ge\ 2(1+\varepsilon)\ =\ 2+2\varepsilon.
    \]
    Because $v_i(g)=1+\varepsilon>0$ for every $i\ge 2$ and every $g\in S$, these are precisely
    the relevant EFX inequalities (with envying agent $i$ and envied agent $1$),
    and we obtain $v_i(A_i) < v_i(A_1\setminus\{g\})$, a contradiction.
    Therefore $\gamma\in\{1,2\}$ for any EFX allocation with positive $p$-mean welfare.
    
    With $\gamma\in\{1,2\}$, the remainder of the argument bounding $W_0(\mathcal{A})$ (and hence
    bounding $W_p(\mathcal{A})$ for all $p\le 0$ via $W_p(\mathcal{A})\le W_0(\mathcal{A})$) applies similarly, so we get the
    same upper bound on $\max_{\mathcal{A}\in\Pi_{\mathrm{EFX}}(G)} W_p(\mathcal{A})$ as in the $\mathrm{EFX_0}$ case.
    Combining this with $\OPT_p(\mathcal{I}) \ge 1$, gives us the same lower bound:
    \[
    \mathrm{PoEFX}^\mathcal{I}(p,c) \ge
    \frac{1}{\max_{\gamma\in\{1,2\}}
    \left(\frac{\gamma}{c+1}\cdot\left(1+\frac{c+1-\gamma}{n-1}\right)^{n-1}\right)^{1/n}}.
    \]
    Again, taking the supremum over all instances with $m = n+c$, we get the bound on $\mathrm{PoEFX}(p,c)$ as desired.
    
\subsection{Proof of Theorem~\ref{thm:pleq0_upperbound}}
Let $\mathcal{A}^{\mathrm{opt}} = (A^{\mathrm{opt}}_1,\ldots,A^{\mathrm{opt}}_n)$ be an allocation maximizing $W_p$ over $\Pi(G)$.
    Write
$\OPT_p(\mathcal{I}) := W_p(\mathcal{A}^{\mathrm{opt}}) = \max_{\mathcal{A}\in \Pi(G)} W_p(\mathcal{A})$.
 
 If $\OPT_p(\mathcal{I})=0$, then since $p\le 0$, we have $W_p(A)=0$ for every allocation $\mathcal{A}$. 
 Hence both the numerator and denominator are $0$, and therefore $\mathrm{PoEFX}^\mathcal{I}(p,c)=\mathrm{PoEFX}_0^\mathcal{I}(p,c)=0$.
This case includes $m<n$ (i.e., $c<0$), since some agent must receive $\varnothing$ and thus have utility $0$.
Therefore, assume $\OPT_p(\mathcal{I})>0$. 
Then every agent must receive strictly positive utility, so in particular $m\ge n$ (equivalently $c\ge 0$) and $|\mathcal{A}^{\mathrm{opt}}_i|\ge 1$ for all $i\in N$.

Consequently,
    \begin{equation*}
        \sum_{i \in N} |A^{\mathrm{opt}}_i| = n + c \implies \max_{i \in N} |A_i^{\mathrm{opt}}| \leq 1 + c.
    \end{equation*}
    For each $i \in N$, let $g_i \in A_i^{\mathrm{opt}}$ be a single item of maximum value to agent~$i$ inside $A_i^{\mathrm{opt}}$, and define $b_i := v_i(g_i)$.
    By additivity and maximality, for every $i \in N$, we have that
    \begin{equation} \label{eqn:uppbound_pleq0_b_i}
        b_i \geq \frac{v_i(A_i^{\mathrm{opt}})}{|A_i^{\mathrm{opt}}|}.
    \end{equation}
    Then, define
    \begin{equation*}
        \widehat{W}_p = 
        \begin{cases}
            \left(\frac{1}{n} \sum_{i \in N} b_i^p \right)^{1/p} & \text{ for } p < 0, \\
            \left( \Pi_{i \in N} b_i \right)^{1/n} & \text{ for } p = 0.
        \end{cases}
    \end{equation*}
    Using (\ref{eqn:uppbound_pleq0_b_i}) and the fact that $f(x) = x^p$ is a decreasing function for all $p < 0$ (and $\log$ is increasing), we get
    \begin{equation*}
        \widehat{W}_p \geq
        \begin{cases}
            \frac{W_p(\mathcal{A}^{\mathrm{opt}})}{\max_{i \in N} |A_i^{\mathrm{opt}}|} & \text{ for } p < 0, \\
            \frac{W_0(\mathcal{A}^{\mathrm{opt}})}{\left( \Pi_{i \in N} |A_i^{\mathrm{opt}}|\right)^{1/n}} & \text{ for } p = 0.
        \end{cases}
    \end{equation*}

    Consider the partial allocation in which each agent $i \in N$ receives only $\{g_i\}$. This allocation is EFX$_0$ because every nonempty bundle is a singleton (and dropping its only item leaves the empty set).
    The key observation here is that we can allocate the remaining $c$ items in such a way that $\max_{\mathcal{A} \in \Pi_\textsf{EFX$_0$}(G)} W_p(\mathcal{A}) \geq \widehat{W}_p$.
    This is guaranteed by the constructive existence proof of \citet{mahara2023extension}  for $m \leq n+3$:
    they proved that for general (and hence additive) valuations, there exists an EFX$_0$ allocation when $m \leq n+3$ and the proof is incremental, adding items one-by-one while maintaining the EFX property via a lexicographic potential.
    Thus, starting from \emph{any} EFX$_0$ partial allocation (in particular, the one we constructed), we can augment it by inserting each of the remaining items one at a time while preserving EFX$_0$ until all goods are allocated.

    We can then apply this augmentation to the singleton EFX allocation $(\{g_i\})_{i \in N}$. Since the $p$-mean welfare for $p\le 0$ is nondecreasing in every coordinate, and because the EFX$_0$-restoring exchanges in the incremental construction can be taken as envy-cycle eliminations that do not decrease any agent's value, the utility profile is coordinatewise nondecreasing throughout the augmentation, and hence the $p$-mean welfare never decreases.
    Thus, there exists such a complete EFX allocation $\widehat{\mathcal{A}}$ such that 
    \begin{equation*}
        W_p(\widehat{\mathcal{A}}) \geq \widehat{W}_p.
    \end{equation*}
    Combining the above equations, we get that for $p < 0$,
    \begin{align*}
        \max_{\mathcal{A} \in \Pi_{\mathrm{EFX}_0}(G) }W_p(\mathcal{A}) & \geq W_p(\widehat{\mathcal{A}}) \geq \widehat{W}_p  \geq \frac{W_p(\mathcal{A}^{\mathrm{opt}})}{\max_{i \in N} |A_i^{\mathrm{opt}}|} \geq \frac{W_p(\mathcal{A}^{\mathrm{opt}})}{1+c}.
    \end{align*}
    Thus, 
    \begin{equation*}
        \max_{A\in \Pi_{\mathrm{EFX}_0}(G)} W_p(\mathcal{A}) \ge \frac{\OPT_p(\mathcal{I})}{1+c},
    \end{equation*}
    and therefore,
    \begin{equation*}
        \mathrm{PoEFX}_0^\mathcal{I}(p,c) := \frac{\OPT_p(\mathcal{I})}{\max_{\mathcal{A}\in \Pi_{\mathrm{EFX_0}}(G)} W_p(\mathcal{A})}
        \leq 1+c.
    \end{equation*}
    Taking the supremum over instances with $m=n+c$ gives $\mathrm{PoEFX}_0(p,c)\leq 1+c$.

    For $p = 0$, note that $\Pi_{i \in N} |A_i^{\mathrm{opt}}| \leq \left( \frac{n+c}{n} \right)^n$, giving us
    \begin{equation*}
        \max_{\mathcal{A} \in \Pi_{\mathrm{EFX}_0}(G)} W_0(\mathcal{A}) \geq \frac{W_0(\mathcal{A}^{\mathrm{opt}})}{(n+c)/n} = \frac{n}{n+c} \cdot W_0(\mathcal{A}^{\mathrm{opt}}).
    \end{equation*}
    Consequently,
    \begin{equation*}
        \max_{\mathcal{A}\in \Pi_{\mathrm{EFX}_0}(G)} W_0(\mathcal{A}) \ge \frac{n}{n+c} \OPT_0(\mathcal{I}),
    \end{equation*}
    and therefore,
    \begin{equation*}
        \mathrm{PoEFX}_0^\mathcal{I}(0,c) := \frac{\OPT_0(\mathcal{I})}{\max_{\mathcal{A}\in \Pi_{\mathrm{EFX}_0}(G)} W_0(\mathcal{A})}
        \le \frac{n+c}{n} = 1+\frac{c}{n}.
    \end{equation*}
    Taking the supremum over instances with $m=n+c$ gives $\mathrm{PoEFX}_0(0,c)\leq 1+c/n$.

    Since $\Pi_{\mathrm{EFX}_0}(G) \subseteq \Pi_{\mathrm{EFX}}(G)$, the same upper bounds also hold for $\mathrm{PoEFX}(p,c)$ and $\mathrm{PoEFX}(0,c)$.

\section{Omitted Proofs in Section~\ref{sec:efx+po}}
\subsection{Proof of Theorem~\ref{thm:PO_EFX}}
    We first prove the NP-hardness result for EFX + PO.
    We begin with the following lemma.
    \begin{lemma} \label{lem:pad-zero-agent}
        Let $\mathcal{I}=(N,G,\mathbf{v})$ be an instance in which every good $g\in G$ has $v_i(g)>0$ for at least one
        agent $i\in N$. 
        Let $\mathcal{I}'=(N\cup\{d\},G,\mathbf{v}')$ be the instance obtained from $\mathcal{I}$ by adding a new agent $d$
        with $v'_d(g)=0$ for all $g\in G$ and $v'_i(\cdot)=v_i(\cdot)$ for all $i\in N$.
        Then $\mathcal{I}$ admits an allocation that is both \emph{EFX} and \emph{PO} if and only if $\mathcal{I}'$ does.
    \end{lemma}
    
    \begin{proof}
    ($\Rightarrow$)
    Let $ \mathcal{A}= (A_1,\dots, A_n)$ be an allocation in $\mathcal{I}$ that is EFX and PO.
    Define $\mathcal{A}'= (A'_1,\dots, A'_n, A'_d)$ in $\mathcal{I}'$ by $A'_i := A_i$ for $i\in N$ and $A'_d:=\varnothing$.
    
    We first show EFX.
    All constraints among original agents are unchanged.
    If $j=d$, then $A'_j=\varnothing$, so the EFX\ condition is vacuous.
    If $i=d$, then $v'_d(g)=0$ for every good, hence there is no $g\in A'_j$ with $v'_d(g)>0$,
    so the EFX\ condition is again vacuous.
    
    Next, we show PO.
    Suppose for contradiction that a complete allocation $\mathcal{B}'$ Pareto-dominates $\mathcal{A}'$ in $\mathcal{I}'$. 
    Agent $d$ cannot be strictly improved (since $v'_d \equiv 0$), thus, some $i \in N$ must be strictly improved.
    
    Moreover, we may assume without loss of generality that $B'_d=\varnothing$: if $g \in B'_d$, then by assumption there exists some agent $k \in N$ with $v_k(g)>0$. 
    Reassign $g$ from $d$ to such an agent $k$. This does not decrease any agent's utility (since $v'_d(g)=0$ and all valuations are nonnegative), and weakly increases agent $k$'s utility. 
    Repeating this for all goods in $B'_d$ gives us a complete allocation $\widehat{\mathcal{B}}$ with $\widehat{\mathcal{B}}_d=\varnothing$ that still Pareto-dominates $\mathcal{A}'$.
    
    Therefore, restricting $\widehat{\mathcal{B}}$ to the original agents $N$ gives a complete allocation of $G$ that Pareto-dominates $\mathcal{A}$ in $I$, contradicting the Pareto-optimality of $\mathcal{A}$. 
    Thus, $\mathcal{A}'$ is PO.
    
    ($\Leftarrow$)
    Let $\mathcal{A}'$ be a complete allocation in $\mathcal{I}'$ that is EFX and PO.
    By the same PO argument as above, $A'_d=\varnothing$ must hold (otherwise some original good is assigned to $d$,
    and moving it to an original agent who values it positively gives us a Pareto-improvement).
    Thus, restricting $\mathcal{A}'$ to the original agents gives us a complete allocation $\mathcal{A}$ of $G$ in $\mathcal{I}$.
    
    All inequalities between original agents are identical in $\mathcal{A}$ and $\mathcal{A}'$, and thus EFX trivially holds
    
    Next, we show PO. 
    If $\mathcal{A}$ were not PO in $\mathcal{I}$, let $\mathcal{B}$ be a Pareto-improvement of $\mathcal{A}$ in $\mathcal{I}$. 
    Extending $\mathcal{B}$ to $\mathcal{I}'$ by giving $d$ the empty bundle results in a Pareto-improvement of $\mathcal{A}'$ in $\mathcal{I}'$, contradicting PO of $\mathcal{A}'$.
    \end{proof}

    We then prove the following padding gadget that increases $m-n$ by $1$.
    \begin{lemma}\label{lem:pad-private-pair}
        Let $\mathcal{I}=(N,G,\mathbf{v})$ be an instance in which every good $g\in G$ has $v_i(g)>0$ for at least one
        agent $i\in N$.
        Let $\mathcal{I}^+=(N\cup\{d\},G\cup\{x,y\},\mathbf{v}^+)$ be obtained by adding a new agent $d$ and two new goods $x,y$
        such that $v^+_d(x)=v^+_d(y)=1$, $v^+_d(g)=0$ for all $g\in G$, and $v^+_i(x)=v^+_i(y)=0$ for all $i\in N$.
        Then $\mathcal{I}$ admits a complete allocation that is both EFX and PO if and only if $\mathcal{I}^+$ does.
    \end{lemma}
    
    \begin{proof}
        ($\Rightarrow$)
        Let $\mathcal{A}=(A_i)_{i\in N}$ be a complete EFX and PO allocation for $\mathcal{I}$.
        Define $\mathcal{A}^+$ for $\mathcal{I}^+$ by setting $A^+_i:=A_i$ for all $i\in N$ and $A^+_d:=\{x,y\}$.
        
        We first prove EFX.
        Among original agents, EFX holds because it held in $\mathcal{A}$ and $x,y$ have value $0$ to all original agents.
        For any original agent $i\in N$ and the new agent $d$, we have $v^+_i(x)=v^+_i(y)=0$,
        so there is no $g\in A^+_d$ with $v^+_i(g)>0$, and the EFX\ condition is vacuous.
        For $i=d$ and any other agent $j$, we have $v^+_d(g)=0$ for all goods $g$ not in $\{x,y\}$,
        so again there is no $g\in A^+_j$ with $v^+_d(g)>0$; EFX\ is vacuous.
        
        Next, we prove PO.
        Suppose for contradiction that a complete allocation $\mathcal{B}^+$ Pareto-dominates $\mathcal{A}^+$ in $\mathcal{I}^+$.
        Agent $d$ already receives all goods it values positively (namely $x$ and $y$), so $d$ cannot be strictly
        improved. Hence some original agent must be strictly improved. 
        Moreover, if $\mathcal{B}^+$ assigns any original good $g \in G$ to agent $d$, then since $v_d(g)=0$ and $g$ has positive value for at least one original agent, we can reassign $g$ from $d$ to such an agent without decreasing any agent's value; thus, w.l.o.g. we may assume $B^+_d \subseteq \{x, y\}$. Hence, restricting $\mathcal{B}^+$ to $N$ gives us an allocation of $G$ that Pareto-improves $\mathcal{A}$, contradicting that $\mathcal{A}$ is PO in $\mathcal{I}$.
        Since original agents value $x,y$ at $0$,
        this strict improvement must come from a different allocation of the original goods $G$ among the original
        agents, while not hurting any original agent. 
        However, that would contradict that $\mathcal{A}$ is PO in $\mathcal{I}$.
        Therefore $\mathcal{A}^+$ is PO.
        
        ($\Leftarrow$)
        Let $\mathcal{A}^+$ be a complete EFX\ and PO allocation in $\mathcal{I}^+$.
        By PO, both $x$ and $y$ must be allocated to agent $d$ (otherwise moving the missing good to $d$ strictly
        improves $d$ and does not decrease any other agent, since every other agent values $x,y$ at $0$).
        Also by the assumption on $G$ and PO, agent $d$ cannot receive any original good $g\in G$ (since some
        original agent values $g$ positively, moving $g$ to that agent is a Pareto-improvement).
        Thus $A^+_d=\{x,y\}$ and all original goods are allocated among original agents.
        Restricting $\mathcal{A}^+$ to $N$ gives a complete allocation $\mathcal{A}$ of $G$ in $\mathcal{I}$.
    
        We first prove EFX.
        For original agents, all relevant inequalities are identical in $\mathcal{A}$ and $\mathcal{A}^+$ because
        the only additional goods are $x,y$ and they have value $0$ to all original agents.
        
        Next, we prove PO.
        If $\mathcal{A}$ were not PO in $\mathcal{I}$, let $\mathcal{B}$ be a Pareto-improvement of $\mathcal{A}$ in $\mathcal{I}$.
        Extend $\mathcal{B}$ to an allocation $\widetilde{\mathcal{B}}$ for $\mathcal{I}^+$ by giving agent $d$ the bundle $\{x,y\}$.
        Then $\widetilde{\mathcal{B}}$ Pareto-dominates $\mathcal{A}^+$ in $\mathcal{I}^+$, contradicting PO of $\mathcal{A}^+$.
    \end{proof}
    
    Then, we reduce from the following NP-hard decision problem: given an instance $\mathcal{I}$, decide whether $\mathcal{I}$ admits a complete allocation that is both EFX and PO.
    This problem is NP-hard for additive valuations (even when each agent has only a constant number of distinct utility values); see, e.g., \citet{GargMurhekar2023FewValues}.
    
    Assume (w.l.o.g.) that every good has positive value for at least one agent.
    This is because we may delete any good $g$ with $v_i(g)=0$ for all agents $i$, since such goods do not affect EFX (they are never `positively valued' by any agent) and do not affect Pareto-optimality; we can always assign them arbitrarily.
    
    We transform $\mathcal{I}$ into an instance $\widehat{\mathcal{I}}$ with an equal number of goods and agents, and then increase the goods-agents gap to exactly $c$.
    
    We first balance to $m=n$ by adding agents/goods.
    \begin{itemize}
        \item If $m\ge n$, apply Lemma~\ref{lem:pad-zero-agent} exactly $(m-n)$ times, adding $(m-n)$ zero agents.
        This gives us an instance $\widehat{\mathcal{I}}$ with $|\widehat{N}|=m$ and $|\widehat{G}|=m$.
        \item If $m<n$, apply Lemma~\ref{lem:pad-private-pair} exactly $(n-m)$ times, adding $(n-m)$ gadgets
        (each gadget adds one agent and two goods). After $(n-m)$ applications we obtain an instance $\widehat{\mathcal{I}}$ with $|\widehat{N}| = n+(n-m)=2n-m$ and $|\widehat{G}| = m+2(n-m)=2n-m$,
        so again $|\widehat{N}|=|\widehat{G}|$.
    \end{itemize}
    By repeated application of Lemma~\ref{lem:pad-zero-agent} and Lemma~\ref{lem:pad-private-pair}, $\mathcal{I}$ admits an EFX and PO allocation if and only if $\widehat{\mathcal{I}}$ does.
    
    Next, we create the gap $c$.
    Starting from $\mathcal{I}_b$ (which satisfies $|G_b|=|N_b|$):
    \begin{itemize}
        \item If $c \ge 0$, apply Lemma~\ref{lem:pad-private-pair} exactly $c$ more times. Each application increases $|G|-|N|$ by $1$, so the resulting instance $\mathcal{I}^{(c)}$ satisfies $|G^{(c)}|-|N^{(c)}| = c$, i.e., $|G^{(c)}| = |N^{(c)}| + c$.
        \item If $c < 0$, apply Lemma~\ref{lem:pad-zero-agent} exactly $(-c)$ more times, adding $(-c)$ zero agents. Each application decreases $|G|-|N|$ by 1, so the resulting instance $I^{(c)}$ satisfies $|G^{(c)}|-|N^{(c)}| = c$, i.e., $|G^{(c)}| = |N^{(c)}| + c$.
    \end{itemize}
    Again by repeated application of Lemma~\ref{lem:pad-zero-agent} and Lemma~\ref{lem:pad-private-pair}, $\mathcal{I}_b$ admits an EFX and PO allocation if and only if $\mathcal{I}^{(c)}$ does.
    
    Note that the mapping described is polynomial-time and preserves yes/no instances.
    The construction adds $\mathcal{O}(|m-n|+|c|)$ agents and goods, and all new valuations are in $\{0,1\}$, so the encoding size grows only polynomially.
    Therefore, deciding whether there exists an EFX and PO allocation is NP-hard even under the promise $|G| = |N| + c$ for any fixed integer $c \leq 3$ (with $m = n + c \geq 1$).

    Next, we prove the $\Sigma_2^P$-completeness result for EFX$_0$ + PO.

    We first show that the problem is in $\Sigma_2^P$. To do so, we nondeterministically guess an allocation $\mathcal{A}=(A_1,\ldots,A_n)$. We can then check whether $\mathcal{A}$ satisfies EFX$_0$ in polynomial time by verifying, for each $i\neq j$ and each $g\in A_j$, that $v_i(A_i)\ge v_i(A_j\setminus\{g\})$.
    To check whether $\mathcal{A}$ is PO, we query an NP oracle with the question:
    \emph{``Is there a complete allocation $\mathcal{B}=(B_1,\ldots,B_n)$ such that $v_i(B_i)\ge v_i(A_i)$ for all $i\in N$, and $v_{i^*}(B_{i^*})>v_{i^*}(A_{i^*})$ for some $i^*\in N$?''}
    We accept if and only if the oracle answers ``no''. Hence the problem is in $\text{NP}^\text{NP}=\Sigma_2^P$.

    Next, we show $\Sigma_2^P$-hardness. We reduce from the problem of deciding whether there exists an allocation that is both EF and PO, which is $\Sigma_2^P$-complete for additive valuations \citep{dekeijzer2009efpo_sigmahard}.
    First, we can without loss of generality delete any good valued $0$ by all agents; this preserves EF and PO existence.
    Next, for each agent $i\in N$, we can also without loss of generality assume that all bundle value differences $v_i(S)-v_i(T)$ are even integers: let $L_i$ be the least common multiple of the denominators of $\{v_i(g): g\in G\}$ and scale $v_i$ by $2L_i$, so that every $v_i(g)$ is an even integer and thus every bundle value (and hence every difference of bundle values) is an even integer.\footnote{Multiplying each agent $i$'s valuation by a positive constant preserves EF and PO: for any two bundles $S,T$, $v_i(S)\ge v_i(T)$ iff $\alpha_i v_i(S)\ge \alpha_i v_i(T)$. Hence the instance is a yes-instance iff the scaled instance is.}

    We now describe the construction of a new instance $\mathcal{I}'=(N',G', \mathbf{v}')$:
    \begin{itemize}
        \item \emph{Agents:} $N'=[n+m]$, where agents $1,\ldots,n$ are the original ones, and agents
        $n+1,\ldots,n+m$ are dummy agents.
        \item \emph{Goods:} $G'=\{g_1,\ldots,g_m,p_1,\ldots,p_n\}$, where $g_1,\ldots,g_m$ are the original goods,
        and for each $i\in [n]$, $p_i$ is the private good for original agent $i$.
        \item \emph{Valuations:} for $i\le n$, set $v'_i(g)=v_i(g)$ for all $g\in G$, $v'_i(p_i)=1$, and
        $v'_i(p_j)=0$ for all $j\neq i$; for $i>n$, set $v'_i(g')=0$ for all $g'\in G'$.
    \end{itemize}

    We will prove that there exists an allocation $\mathcal{A}'$ that is EFX$_0$ and PO for $\mathcal{I}'$ if and only if there exists a complete allocation $\mathcal{A}$ that is EF and PO for $\mathcal{I}$.

    For the `if' direction, suppose there exists an allocation $\mathcal{A}=(A_1,\ldots,A_n)$ for $\mathcal{I}$ that is EF and PO. 
    Define an allocation $\mathcal{A}'$ for $\mathcal{I}'$ by setting $A'_i:=A_i\cup\{p_i\}$ for all $i\le n$ and $A'_i:=\varnothing$ for all $i>n$.

    We first show that $\mathcal{A}'$ is PO.
    Suppose for contradiction that there is a complete allocation $\mathcal{B}'=(B'_1,\ldots,B'_{n+m})$ that Pareto-improves $\mathcal{A}'$ in $\mathcal{I}'$. Consider the restriction of $\mathcal{B}'$ to the original goods: for each $i\le n$ define $B_i:=B'_i\cap G$, and reassign any original good in $G$ that is allocated in $\mathcal{B}'$ to a dummy agent to an arbitrary original agent (this weakly increases that original agent's utility and does not decrease anyone else's utility because dummy agents value all goods at $0$).
    After this reassignment, we obtain an allocation $\widetilde{\mathcal{B}}$ of the original goods among the original agents. 
    For every original agent $i\le n$, we have
    \begin{equation*}
        v_i(\widetilde{B}_i) \geq v_i(B'_i\cap G) \geq v'_i(B'_i)-1 \geq v'_i(A'_i)-1 = v_i(A_i),
    \end{equation*}
    where the second inequality uses that agent $i$ can gain at most $1$ utility from private goods (namely from $p_i$), and the last equality follows from $v'_i(p_i)=1$.
    Moreover, strict Pareto improvement of $\mathcal{B}'$ over $\mathcal{A}'$ implies that some original agent is strictly better in $\mathcal{B}'$, and since private-good utility cannot increase beyond $1$ (already attained in $\mathcal{A}'$), this strict improvement must come from original goods; hence $v_i(\widetilde{B}_i)>v_i(A_i)$ for some $i$.
    Therefore, $\widetilde{\mathcal{B}}$ Pareto-improves $\mathcal{A}$ in $\mathcal{I}$, contradicting that $\mathcal{A}$ is PO. Thus $\mathcal{A}'$ is PO.

    Next, we show that $\mathcal{A}'$ satisfies EFX$_0$.
    All EFX$_0$ inequalities involving a dummy agent hold trivially because every dummy agent values all goods at $0$ (and in $A'$ each dummy agent holds $\varnothing$).
    Thus, consider any two distinct original agents $i,j\in N$ and any good $g\in A'_j$.
    If $g=p_j$, then
    \begin{equation*}
        v'_i(A'_i)=v_i(A_i)+1 \ge v_i(A_j)= v'_i(A'_j\setminus\{p_j\}),
    \end{equation*}
    where the inequality follows because $\mathcal{A}$ is EF (so $v_i(A_i)\ge v_i(A_j)$).
    If $g\neq p_j$ (so $g\in A_j$), then using $v'_i(p_j)=0$,
    \begin{align*}
        v'_i(A'_j\setminus\{g\}) & = v_i(A_j\setminus\{g\}) \leq v_i(A_j) \leq v_i(A_i) < v_i(A_i)+1 = v'_i(A'_i).
    \end{align*}
    Thus $\mathcal{A}'$ is EFX$_0$.

    Next, we prove the `only if' direction.
    Suppose there exists an allocation $\mathcal{A}'=(A'_1,\ldots,A'_{n+m})$ for $\mathcal{I}'$ that is EFX$_0$ and PO.
    Since $\mathcal{A}'$ is PO, we must have $p_i\in A'_i$ for each $i\in [n]$, otherwise reallocating $p_i$ to agent $i$ strictly increases agent $i$'s utility and does not decrease any other agent's utility (all other agents value $p_i$ at $0$).
    Moreover, since each good in $G'$ has positive value to at least one original agent (after removing goods valued $0$ by all), no PO allocation can assign any good in $G'$ to a dummy agent; hence $A'_i=\varnothing$ for all $i>n$.
    
    Define $A_i:=A'_i\setminus\{p_i\}$ for each $i\in N$. 
    We show that $\mathcal{A}=(A_1,\ldots,A_n)$ is EF.
    Fix any two distinct agents $i,j\in [n]$. 
    Apply the EFX$_0$ inequality to agents $i$ and $j$ with the good $g=p_j\in A'_j$:
    \begin{equation*}
        v'_i(A'_i)\;\ge\; v'_i(A'_j\setminus\{p_j\}).
    \end{equation*}
    Using $v'_i(A'_i)=v_i(A_i)+1$ and $v'_i(A'_j\setminus\{p_j\})=v_i(A_j)$ (since $v'_i(p_j)=0$ for $i\neq j$), we obtain $v_i(A_i)+1\ge v_i(A_j)$.
    Since $v_i(A_i)$ and $v_i(A_j)$ are even integers (by the scaling), this implies $v_i(A_i)\ge v_i(A_j)$. 
    Hence $\mathcal{A}$ is EF.

    Finally, $\mathcal{A}$ must be PO: if not, then there exists an allocation $\mathcal{B}=(B_1,\ldots,B_n)$ of $G$ that Pareto-improves $\mathcal{A}$ in $\mathcal{I}$.
    Extend $\mathcal{B}$ to an allocation $\mathcal{B}'$ of $G'$ in $\mathcal{I}'$ by giving each original agent $i$ the bundle $B_i\cup\{p_i\}$ and giving all dummy agents the empty bundle. 
    Then $\mathcal{B}'$ Pareto-improves $\mathcal{A}'$ in $\mathcal{I}'$, contradicting the fact that $\mathcal{A}'$ is PO.

    This completes the reduction. 
    Note that in our construction $|N'|=n+m$ and $|G'|=m+n$, so $|G'|=|N'|$; in particular, $|G'|\le |N'|$. 
    Therefore, the problem is $\Sigma_2^P$-hard already on instances with $m=n$, and hence also on instances with $m\le n$.

    Next, we describe an extension to instances where $m = n + c$.
    Fix any constant $c\in\{1,2,3\}$. Starting from the instance $\mathcal{I}'=(N',G',v')$ above (where $|G'|=|N'|$), we build an instance $\mathcal{I}^{(c)}=(N^{(c)},G^{(c)},v^{(c)})$ with $|G^{(c)}|=|N^{(c)}|+c$ as follows.
    Let $D:=\{d_1,\ldots,d_c\}$ be $c$ new (dummy) agents, and let $H:=\{x_1,y_1,\ldots,x_c,y_c\}$ be $2c$ new goods. Define
    \begin{equation*}
        N^{(c)} := N'\cup D,\qquad G^{(c)} := G'\cup H,
    \end{equation*}
    so $|N^{(c)}|=|N'|+c$ and $|G^{(c)}|=|G'|+2c=(|N'|+c)+c=|N^{(c)}|+c$.
    Set valuations as follows:
    \begin{itemize}
        \item For every agent $a\in N'$ and every $g\in G'$, let $v^{(c)}_a(g):=v'_a(g)$, and for every $h\in H$ let $v^{(c)}_a(h):=0$.
        \item For each $t\in [c]$, let $v^{(c)}_{d_t}(x_t)=v^{(c)}_{d_t}(y_t)=1$, and $v^{(c)}_{d_t}(g)=0$ for all $g\in G^{(c)}\setminus\{x_t,y_t\}$.
    \end{itemize}

    We claim that $\mathcal{I}'$ admits an EFX$_0$ and PO allocation if and only if $\mathcal{I}^{(c)}$ admits an EFX$_0$ and PO allocation.
    Indeed, given any complete EFX$_0$ and PO allocation $\mathcal{A}'$ for $\mathcal{I}'$, extend it to $\mathcal{I}^{(c)}$ by additionally allocating
    $\{x_t,y_t\}$ to agent $d_t$ for each $t\in[c]$; EFX$_0$ holds because agents in $N'$ value all goods in $H$ at $0$ and each $d_t$ values all goods outside $\{x_t,y_t\}$ at $0$, and PO holds because any Pareto improvement would have to leave each $d_t$'s utility unchanged (so in particular keep $\{x_t,y_t\}$ with $d_t$), and would then induce a Pareto improvement of $\mathcal{A}'$ on the subinstance $(N',G')$.
    Conversely, let $\mathcal{A}^{(c)}$ be a complete EFX$_0$ and PO allocation for $\mathcal{I}^{(c)}$. 
    By PO, for each $t\in[c]$ we must have $x_t,y_t\in \mathcal{A}^{(c)}_{d_t}$, otherwise moving the missing good to $d_t$ strictly increases $d_t$'s utility and does not decrease any other agent's utility (no other agent values that good).
    Moreover, since every good in $G'$ has positive value to some agent in $N'$ while each $d_t$ values all of $G'$ at $0$, PO also implies that no good in $G'$ is allocated to any $d_t$.
    Therefore, restricting $\mathcal{A}^{(c)}$ to the agents in $N'$ gives us a complete allocation $A'$ of $G'$.
    This restriction remains EFX$_0$ (valuations on $G'$ are unchanged for agents in $N'$), and if it were not PO in $\mathcal{I}'$, a Pareto improvement in $\mathcal{I}'$ could be extended by keeping each $d_t$'s bundle $\{x_t,y_t\}$ fixed, contradicting PO of $\mathcal{A}^{(c)}$.
    
    Thus, the reduction above can be followed by this padding step, giving us $\Sigma_2^P$-hardness even when $m=n+c$ for each fixed $c\in\{1,2,3\}$. 

    To handle any fixed $c < 0$, start from the constructed instance $\mathcal{I}'$ (which satisfies $|G'| = |N'|$) and add $-c$ new dummy agents with valuation $0$ for every good, and no new goods. 
    The same argument as in Lemma~\ref{lem:pad-zero-agent} (and it continues to hold under EFX$_0$) shows that the resulting instance admits an EFX$_0$ and PO allocation if and only if $\mathcal{I}'$ does. The new instance satisfies $m = n + c$, establishing $\Sigma_2^P$-completeness for fixed negative $c$ as well.
    Together with membership in $\Sigma_2^P$, this gives $\Sigma_2^P$-completeness for these settings.

\section{Omitted Proofs in Section~\ref{sec:other_settings}}
\subsection{Proof of Theorem~\ref{thm:nmu-efx-pmean}}
    Under NMU and additivity we have $v_i(g)>0$ for all $i\in N$ and $g\in G$, thus, EFX and EFX$_0$ coincide.
    We therefore work with EFX$_0$ throughout.
    
    \paragraph{Case $m<n$ (i.e., $c<0$).}
    Every allocation has at least one empty bundle. 
    If $p\le 0$, then $W_p(\mathcal{A})=0$ for every
    allocation $\mathcal{A}$ and in this case, output any complete allocation that assigns each good to a distinct agent (possibly since $m < n$); this allocation is EFX$_0$ (and hence EFX) because all nonempty bundles are singletons).
    If $p>0$, note that in any EFX$_0$ allocation, every bundle must have size at most $1$: indeed, since $m<n$
    there exists an empty agent $i$, and if some agent $j$ had $|A_j|\ge 2$, then for any $g\in A_j$, we would have
    $A_j\setminus\{g\}\neq\varnothing$ and thus $v_i(A_j\setminus\{g\})>0=v_i(A_i)$, contradicting EFX$_0$.
    As such, EFX$_0$ allocations are exactly injective assignments of goods to agents (a matching of size $m$).
    For $p>0$, maximizing the $p$-mean welfare is equivalent to maximizing $\sum_{i\in N} v_i(A_i)^p$, and the unmatched agents contribute $0$ to this sum. 
    As a result, we can compute an optimal EFX$_0$ allocation by solving a maximum weight matching (which matches all goods) of size $m$ in the bipartite graph $N\times G$ with edge weights $w(i,g):=v_i(g)^p$.
    This is known to be polynomial-time solvable.
    
    \paragraph{Case $m\ge n$ (i.e., $c\ge 0$).}
    We first show a key structural consequence of NMU.
    
    \medskip\noindent\textbf{Claim.}
    If $m\ge n$, then every EFX$_0$ allocation $\mathcal{A}=(A_1,\dots,A_n)$ satisfies $A_i\neq\varnothing$ for all $i\in N$.
    
    \smallskip\noindent\emph{Proof of Claim.}
    Suppose for contradiction that $A_i=\varnothing$ for some $i\in N$.
    Since $m\ge n$, among the remaining $n-1$ bundles there are $m$ goods total, so some agent $j\neq i$ has $|A_j|\ge 2$.
    Pick any $g\in A_j$. 
    Then $A_j\setminus\{g\}\neq\varnothing$, and by NMU (equivalently $v_i(h)>0$ for all goods $h$)
    we have $v_i(A_j\setminus\{g\}) > 0 = v_i(A_i)$, violating EFX$_0$.
    \hfill$\square$
    
    \medskip
    Now, fix $m=n+c$ with $c\in\{0,1,2,3\}$.
    We call an agent \emph{heavy} if she receives at least two goods. 
    For an allocation $\mathcal{A}$, define
    \begin{equation*}
        H(\mathcal{A}):=\{i\in N:\ |A_i|\ge 2\} \quad\text{and}\quad k:=|H(\mathcal{A})|.
    \end{equation*}
    By Claim, every agent receives at least one good in any EFX$_0$ allocation, and thus the total number of ``extra'' goods beyond one per agent is exactly
    \begin{equation*}
        \sum_{i\in N} (|A_i|-1)=m-n=c.
    \end{equation*}
    Each heavy agent contributes at least $1$ to this sum, so $k\le c$. 
    Moreover, since bundles are disjoint,
    \begin{align*}
        \Bigl|\bigcup_{i\in H(A)} A_i\Bigr|
        =\sum_{i\in H(A)} |A_i|
        & =\sum_{i\in H(A)}(1+(|A_i|-1)) \\
        & = k+c \le 2c.
    \end{align*}
    In particular, when $c\le 3$, the union of all non-singleton bundles contains at most $2c\le 6$ goods.

    We now describe the algorithm.
    Since EFX$_0$ allocations exist for $m\le n+3$, the optimum
    \[
        \OPT:=\max_{\mathcal{A}\in \Pi_{\mathrm{EFX}_0}(G)} W_p(\mathcal{A})
    \]
    is well-defined. 

    We first enumerate the heavy part (i.e., allocation to the heavy agents).
    Then, for each $k\in\{0,1,\dots,c\}$ (with $k=0$ only possible if $c=0$), enumerate:
    \begin{enumerate}
    \item a set $H\subseteq N$ of heavy agents with $|H|=k$;
    \item a set $S\subseteq G$ of goods with $|S|=k+c$;
    \item a labeled partition $(B_i)_{i\in H}$ of $S$ such that $|B_i|\ge 2$ for all $i\in H$.
    \end{enumerate}
    Let $L:=N\setminus H$ and $R:=G\setminus S$. Since $|G|=n+c$ and $|S|=k+c$, we have $|R|=n-k=|L|$.
    Any completion of this partial allocation must therefore assign exactly one good from $R$ to each light agent in $L$.
    
    For each heavy agent's bundle $B_i$ and each agent $x\in N$, define
    \[
    \tau_x(B_i):=\max_{g\in B_i} v_x(B_i\setminus\{g\}), 
    \tau_\ell:=\max_{i\in H}\tau_\ell(B_i)\ \text{ for }\ \ell\in L,
    \]
    with the convention $\tau_\ell:=0$ if $H=\varnothing$.
    Consider any bijection $\pi:L\to R$ and the completion $\mathcal{A}^\pi$ that gives each heavy agent $i\in H$ the bundle $B_i$ and each light agent $\ell\in L$ the singleton $\{\pi(\ell)\}$.
    Then $\mathcal{A}^\pi$ is EFX$_0$ if and only if:
    \begin{enumerate}[(i)]
        \item for all distinct $a,i\in H$, $v_a(B_a)\ge \tau_a(B_i)$;
        \item for all $\ell\in L$, $v_\ell(\pi(\ell))\ge \tau_\ell$.
    \end{enumerate}
    Indeed, EFX$_0$ constraints where the envied bundle is a singleton are always satisfied since removing its only
    good leaves $\varnothing$. 
    Thus only envy toward heavy agents' bundles matters.
    Fix $i\in H$ and $x\in N$. 
    The EFX$_0$ constraints from $x$ toward $i$ are
    $v_x(A^\pi_x)\ge v_x(B_i\setminus\{g\})$ for all $g\in B_i$, equivalently $v_x(A^\pi_x)\ge\tau_x(B_i)$.
    If $x=a\in H$ then $A^\pi_x=B_a$, giving (i); if $x=\ell\in L$ then $A^\pi_x=\{\pi(\ell)\}$ and we obtain (ii) after taking the maximum over $i\in H$. 
    The converse follows by the same case analysis.
    
    We discard the current heavy part if condition (i) fails.
    
    Next, build a bipartite graph on $L\times R$ containing an edge $(\ell,g)$ iff $v_\ell(g)\ge \tau_\ell$.
    A perfect matching $\pi:L\to R$ corresponds to a feasible completion satisfying (ii), hence giving an EFX$_0$ allocation.
    
    Among all feasible perfect matchings, we optimize for the $p$-mean welfare. 
    Let $u_i$ denote agent $i$'s utility in the completed allocation.
    Within a fixed iteration, the heavy agent utilities $u_i=v_i(B_i)$ are constant, so optimizing for $p$-mean welfare reduces to optimizing over the light assignments only:
    \begin{itemize}
        \item If $p>0$, maximize $\sum_{\ell\in L} v_\ell(\pi(\ell))^p$ (maximum weight perfect matching with weights
        $w(\ell,g)=v_\ell(g)^p$).
        \item If $p=0$, maximize $\sum_{\ell\in L}\log v_\ell(\pi(\ell))$ (maximum weight perfect matching with weights
        $w(\ell,g)=\log v_\ell(g)$).
        \item If $p<0$, minimize $\sum_{\ell\in L} v_\ell(\pi(\ell))^p$ (minimum cost perfect matching with costs
        $c(\ell,g)=v_\ell(g)^p$).
    \end{itemize}
    These are equivalent transformations because the function $f(x)= x^{1/p}$ is strictly increasing for $p>0$ and strictly decreasing for $p<0$, and for $p=0$ we use the standard Nash welfare logarithmic objective.
    We solve the corresponding assignment problem (e.g., using the Hungarian algorithm or min-cost flow). 
    If no perfect matching exists, discard this heavy part. Otherwise, record the best completion for this heavy part, and finally output the best recorded allocation over all iterations.

    We now prove correctness.
    Let $\mathcal{A}^*$ be an EFX$_0$ allocation maximizing the $p$-mean welfare. 
    By the Claim, all bundles in $\mathcal{A}^*$ are nonempty.
    Then, $|H(\mathcal{A}^*)| \le c$, and let $S^*:=\bigcup_{i\in H(\mathcal{A}^*)} A^*_i$.
    As shown above, $|S^*|=|H(\mathcal{A}^*)|+c\le 2c$, and the heavy part $(H(\mathcal{A}^*),S^*,(A^*_i)_{i\in H(\mathcal{A}^*)})$ appears in the enumeration earlier. 
    In that iteration, the light agent assignment induced by $\mathcal{A}^*$ satisfies (ii) by definition of $\tau_\ell$, and thus corresponds to a perfect matching. 
    Then (i) holds because $\mathcal{A}^*$ is EFX$_0$.
    The matching subroutine returns an optimal completion for this fixed heavy part, thus giving us a welfare at least $W_p(\mathcal{A}^*)$.
    Taking the best over all iterations returns an EFX$_0$ allocation attaining $\OPT$.
    Finally, since EFX$_0$ and EFX coincide under NMU, the returned allocation is also EFX and is optimal among EFX allocations.

    Finally, we prove the running time of this algorithm.
    For fixed $c\le 3$, we have $|S|=k+c\le 2c\le 6$. 
    Thus, there are $\mathcal{O}(n^3)$ choices for $H$ and $\mathcal{O}(m^6)$ choices for $S$, and the number of labeled partitions of a set of size at most $6$ into at most $3$ parts of size at least $2$ is constant.
    Each iteration solves one assignment problem on $|L|= \mathcal{O}(n)$ vertices in polynomial time. Consequently, the overall running time is polynomial in the input size.

\subsection{Proof of Theorem~\ref{thm:other_fixedc}}
    For notational simplicity, we denote $\OPT_p(\mathcal{I}) := \max_{\mathcal{A}\in \Pi(G)} W_p(\mathcal{A})$ and $\OPT^{\mathrm{EFX}}_p(\mathcal{I}) := \max\{W_p(\mathcal{A}) : \mathcal{A} \in \Pi(G)\text{ is EFX}\}$,
    with the convention $\OPT^{\mathrm{EFX}}_p(\mathcal{I}):=-\infty$ if no EFX allocation exists.
    We describe an XP (with respect to $c$) algorithm that computes $\OPT_p(\mathcal{I})$ and $\OPT^{\mathrm{EFX}}_p(\mathcal{I})$ and then compares the two.

    We describe an XP-time algorithm (parameter $c$) that computes $\OPT_p(\mathcal{I})$ and the best $p$-mean welfare attainable by an EFX allocation, and then compares the two.
    
    Let $\Gamma=(N,G;E)$ be the bipartite graph with $(i,g)\in E$ iff $v_i(g)>0$.
    If $\Gamma$ has no matching of size $n$, then no allocation can give every agent strictly positive utility:
    indeed, if an allocation $\mathcal{A}$ satisfied $v_i(A_i)>0$ for all $i$, then for each agent $i$ we could pick some $g_i\in A_i$ with $v_i(g_i)>0$; disjointness of bundles would imply that $\{(i,g_i)\}_{i\in N}$ is a matching of size $n$, a contradiction.
    When $p\le 0$, our convention is that $W_p(A)=0$ for any allocation $\mathcal{A}$ with $v_i(A_i)=0$ for some agent $i$, and hence in the above degenerate case $\OPT_p(\mathcal{I})=0$.
    Since the theorem assumes $\OPT_p(\mathcal{I})>0$, we are in the nondegenerate case where $\Gamma$ has a matching of size $n$; in particular, every allocation attaining $\OPT_p(\mathcal{I})$ gives every agent strictly positive utility, and thus every agent receives at least one good.
    This also implies $m\ge n$ (equivalently, $c\ge 0$).
    
    Fix any allocation $\mathcal{A}=(A_1,\dots,A_n)$ in which every agent receives at least one good.
    Define the heavy agents and heavy goods by
    \[
        H(\mathcal{A}) := \{i\in N : |A_i|\ge 2\} \text{ and }
        S(\mathcal{A}) := \bigcup_{i\in H(\mathcal{A})} A_i.
    \]
    Since $m=n+c$ and $|A_i|\ge 1$ for all $i$, the number of ``extra'' goods beyond one per agent is
    \[
    \sum_{i\in N}(|A_i|-1) = m-n = c.
    \]
    Light agents (those with $|A_i|=1$) contribute $0$ to this sum, while each heavy agent contributes at least $1$, so $|H(\mathcal{A})|\le c$.
    Moreover, because bundles are disjoint,
    \begin{align*}
        |S(\mathcal{A})| =
        \sum_{i\in H(A)}|A_i|
        & = \sum_{i\in H(\mathcal{A})}(1+(|A_i|-1)) \\
        & = |H(A)| + \sum_{i\in H(A)}(|A_i|-1) \\
        & = |H(A)|+c \le 2c.        
    \end{align*}

    In particular, every allocation with $W_p(\mathcal{A})>0$ has at most $c$ heavy agents and at most $2c$ heavy goods.
    
    For a bundle $B\subseteq G$ and an agent $x\in N$, define the EFX threshold
    \begin{equation*}
        \tau_x(B) := \max_{g\in B: v_x(g)>0} v_x(B\setminus\{g\}),
    \end{equation*}
    with the convention $\tau_x(B):=0$ if $\{g\in B : v_x(g)>0\}=\varnothing$.
    For EFX$_0$, drop the filter $v_x(g)>0$ and take the maximum over all $g\in B$.
    
    Next, we show a characterization of EFX completions for a fixed assignment of goods to the heavy agents.
    Fix sets $H\subseteq N$ and $S\subseteq G$ with $|S|=|H|+c$, and a labeled partition $(B_i)_{i\in H}$ of $S$ such that $|B_i|\ge 2$ for all $i\in H$.
    Let $L:=N\setminus H$ and $R:=G\setminus S$; then $|L|=|R|$.
    For each $\ell\in L$, define $\tau_\ell:=\max_{i\in H}\tau_\ell(B_i)$, and set $\tau_\ell:=0$ if $H=\varnothing$.
    
    Consider a completion $\mathcal{A}^\pi$ that assigns each heavy agent $i\in H$ the bundle $B_i$, and assigns each light agent $\ell\in L$ a singleton $\{\pi(\ell)\}$ via a bijection $\pi:L\to R$.
    We claim that $\mathcal{A}^\pi$ is EFX if and only if:
    \begin{enumerate}
    \item[(i)] for all distinct $a,i\in H$, we have $v_a(B_a)\ge \tau_a(B_i)$; and
    \item[(ii)] for all $\ell\in L$, we have $v_\ell(\pi(\ell))\ge \tau_\ell$.
    \end{enumerate}
    
    \emph{Proof of the claim.}
    The EFX constraints are: for all agents $x,y\in N$ and all goods $g\in A^\pi_y$ with $v_x(g)>0$,
    \[
    v_x(A^\pi_x) \;\ge\; v_x(A^\pi_y\setminus\{g\}).
    \]
    If $y\in L$, then $A^\pi_y$ is a singleton, so $A^\pi_y\setminus\{g\}=\varnothing$ and the inequality becomes $v_x(A^\pi_x)\ge 0$, which holds by nonnegativity.
    Thus only constraints with $y\in H$ matter, i.e., where the compared bundle is some heavy bundle $B_i$.
    
    Fix $i\in H$ and $x\in N$.
    The family of constraints over all $g\in B_i$ with $v_x(g)>0$ is equivalent to the single inequality
    \[
    v_x(A^\pi_x) \;\ge\; \max_{g\in B_i:\,v_x(g)>0} v_x(B_i\setminus\{g\}) \;=\; \tau_x(B_i).
    \]
    If $x=a\in H$, then $A^\pi_x=B_a$, so these inequalities for all $i\in H$ are exactly condition (i).
    If $x=\ell\in L$, then $A^\pi_x=\{\pi(\ell)\}$ and $v_\ell(A^\pi_x)=v_\ell(\pi(\ell))$, so requiring the above for all $i\in H$ is equivalent to
    $v_\ell(\pi(\ell))\ge \max_{i\in H}\tau_\ell(B_i)=\tau_\ell$, which is condition (ii). \hfill$\square$

    Now, we first compute $\OPT_p(\mathcal{I})$.
    We enumerate all possible heavy parts of an allocation with positive welfare.
    For each integer $k\in\{0,1,\dots,\min(c,n)\}$, enumerate:
    \begin{enumerate}
    \item a set $H\subseteq N$ with $|H|=k$;
    \item a set $S\subseteq G$ with $|S|=k+c$;
    \item a labeled partition $(B_i)_{i\in H}$ of $S$ such that $|B_i|\ge 2$ for all $i\in H$.
    \end{enumerate}
    Note that if $c>0$ and every agent receives at least one good (equivalently, $W_p(\mathcal{A})>0$ for $p\le 0$), then necessarily $k\ge 1$; thus the $k=0$ case is only relevant when $c=0$.
    
    For a fixed triple $(H,S,(B_i))$, let $L:=N\setminus H$ and $R:=G\setminus S$, so $|L|=|R|$.
    Discard the triple if $v_i(B_i)=0$ for some $i\in H$ (since then every completion has welfare $0$ under our convention for $p\le 0$).
    
    Any completion with positive welfare must assign each $\ell\in L$ exactly one good from $R$, hence corresponds to a bijection $\pi:L\to R$ with $v_\ell(\pi(\ell))>0$ for all $\ell\in L$.
    For fixed heavy bundles, the heavy-agent contributions to the objective are constant, so optimizing the $p$-mean welfare over such completions reduces to an assignment problem on the bipartite graph $L\times R$:
    \begin{itemize}
    \item If $p=0$ (Nash welfare), maximize $\sum_{\ell\in L}\log v_\ell(g)$ over perfect matchings with $v_\ell(g)>0$.
    \item If $p<0$ (for any real $p$), maximizing $W_p$ over strictly positive utilities is equivalent to minimizing $\sum_{i\in N} v_i(A_i)^p$ (since the function $f(x) = x^{1/p}$ is decreasing for $p<0$), which for fixed heavy bundles reduces to minimizing $\sum_{\ell\in L} v_\ell(g)^p$ over perfect matchings with $v_\ell(g)>0$. For $p \rightarrow -\infty$, within a fixed heavy part, we maximize $\min_{\ell\in L} v_\ell(\pi(\ell))$, which can be solved via a bottleneck perfect matching.
    \end{itemize}
    We solve the resulting maximum-weight (for $p=0$) or minimum-cost (for $p<0$) perfect matching in polynomial time and keep the best solution over all iterations.
    Let $\mathcal{A}^*$ denote a witness allocation attaining $\OPT_p(\mathcal{I})$.

    Next, we compute the best welfare given an EFX allocation.
    Repeat the same enumeration over $(H,S,(B_i))$.
    Discard if $v_i(B_i)=0$ for some $i\in H$.
    Next, enforce EFX feasibility for envy comparisons only between heavy agents: discard the triple if there exist distinct $a,i\in H$ with
    $v_a(B_a)<\tau_a(B_i)$ (condition (i) above).
    
    For each $\ell\in L$, compute $\tau_\ell:=\max_{i\in H}\tau_\ell(B_i)$ (or $\tau_\ell:=0$ if $H=\varnothing$).
    Build a bipartite graph on $L\times R$ containing an edge $(\ell,g)$ iff
    \[
    v_\ell(g)>0 \quad\text{and}\quad v_\ell(g)\ge \tau_\ell.
    \]
    By the characterization proved above, perfect matchings in this graph correspond exactly to EFX completions of the fixed heavy part in which every light utility is positive.
    Among those perfect matchings, optimize the same objective as in Step~1.
    Because we compare against $\OPT_p(\mathcal{I})>0$, restricting to strictly positive utilities is without loss of generality for the compatibility decision.
    Taking the best value over all iterations gives us $\OPT^{\mathrm{EFX}}_p(\mathcal{I})$; let $\mathcal{A}^{\mathrm{EFX}}$ be a witness allocation if one exists.
    
    Every EFX allocation is an allocation, and thus, $\OPT^{\mathrm{EFX}}_p(\mathcal{I})\le \OPT_p(\mathcal{I})$.
    Therefore, there exists an EFX allocation that is globally $p$-mean optimal if and only if     $\OPT^{\mathrm{EFX}}_p(\mathcal{I})=\OPT_p(\mathcal{I})$.
    If there exists such an EFX allocation, output the witness allocation $\mathcal{A}^{\mathrm{EFX}}$.
    
    We analyze the running time.
    For each $k\le c$, the enumeration has $\binom{n}{k}$ choices of $H$ and $\binom{m}{k+c}$ choices of $S$.
    A labeled partition of a $(k+c)$-element set into $k$ parts can be enumerated in time $(k+c)^{\mathcal{O}(k+c)}\le (2c)^{\mathcal{O}(2c)}$, which is absorbed by $(n+m)^{\mathcal{O}(c)}$ since $n+m\ge c$.
    Each iteration solves one assignment instance on $|L| = \mathcal{O}(n)$ vertices in polynomial time (e.g., Hungarian algorithm in $\mathcal{O}(n^3)$).
    Thus, the overall running time is $(n+m)^{\mathcal{O}(c)}$, i.e., XP in $c$.

    For EFX$_0$, replace each threshold definition by dropping the filter $v_x(g)>0$ (i.e., take the maximum over all $g\in B$), and otherwise run the same enumeration and matching procedure.
    The characterization and the correctness proof carry over immediately, because EFX differs from EFX$_0$ only in whether goods of zero value to the envying agent are quantified over.
\end{document}